\documentclass[reqno,10pt,centertags]{amsart} 
\usepackage{amsmath,amsthm,amscd,amssymb,latexsym,esint,upref,stmaryrd,enumerate,color,verbatim,yfonts,multicol}
\calclayout
\usepackage{color,float}
\usepackage{hyperref}

\usepackage{tikz}
\usetikzlibrary{calc,decorations.markings}
\usepackage{pgfplots}
\pgfplotsset{compat=1.15}

\usepackage{graphicx}

\newcommand*{\mailto}[1]{\href{mailto:#1}{\nolinkurl{#1}}}
\newcommand{\arxiv}[1]{\href{http://arxiv.org/abs/#1}{arXiv:#1}}

\newcommand{\R}{{\mathbb R}}
\newcommand{\N}{{\mathbb N}}
\newcommand{\Z}{{\mathbb Z}}
\newcommand{\C}{{\mathbb C}}

\newcommand{\bbC}{{\mathbb{C}}}

\newcommand{\bbN}{{\mathbb{N}}}

\newcommand{\bbR}{{\mathbb{R}}}

\newcommand{\cA}{{\mathcal A}}
\newcommand{\cB}{{\mathcal B}}

\newcommand{\cD}{{\mathcal D}}

\newcommand{\cF}{{\mathcal F}}

\newcommand{\cH}{{\mathcal H}}

\newcommand{\cZ}{{\mathcal Z}}

\newcommand{\beq}{\begin{align}}
\newcommand{\enq}{\end{align}}

\renewcommand{\a}{\alpha}
\renewcommand{\b}{\beta}
\newcommand{\g}{\gamma}
\renewcommand{\d}{\delta}

\newcommand{\z}{\zeta}

\newcommand{\cAB}{c}

\DeclareMathOperator{\supp}{supp}

\DeclareMathOperator{\dom}{dom}

\newcommand{\Res}{\text{\rm Res}}
\renewcommand{\Re}{\text{\rm Re}}
\renewcommand{\Im}{\text{\rm Im}}
\renewcommand{\ln}{\text{\rm ln}}

\newcommand{\no}{\notag}
\newcommand{\lb}{\label}
\newcommand{\f}{\frac}

\newcommand{\ol}{\overline}
\newcommand{\bs}{\backslash}

\newcommand{\wti}{\widetilde}
\newcommand{\Oh}{O}
\newcommand{\oh}{o}
\newcommand{\hatt}{\widehat} 
\newcommand{\dott}{\,\cdot\,}

\renewcommand{\dot}{\overset{\textbf{\large.}}}

\newcommand{\bi}{\bibitem}

\let\geq\geqslant
\let\leq\leqslant

\newcommand{\lam}{\lambda}

\newcommand{\al}{\a}
\newcommand{\be}{\b}

\newcommand{\Lr}{{L^2((a,b);rdx)}} 

\newcommand{\ACl}{{AC_{loc}((a,b))}}
\newcommand{\Ll}{{L^1_{loc}((a,b);dx)}}

\makeatletter
\def\theequation{\@arabic\c@equation}

\allowdisplaybreaks 
\numberwithin{equation}{section}

\newtheorem{theorem}{Theorem}[section]
\newtheorem{proposition}[theorem]{Proposition}
\newtheorem{lemma}[theorem]{Lemma}

\newtheorem{definition}[theorem]{Definition}
\newtheorem{hypothesis}[theorem]{Hypothesis}
\newtheorem{example}[theorem]{Example}

\newtheorem{assumption}[theorem]{Assumption}
\theoremstyle{remark}
\newtheorem{remark}[theorem]{Remark}

\begin{document}

\title{The spectral $\zeta$-function for quasi-regular Sturm--Liouville operators} 

\author[G.\ Fucci]{Guglielmo Fucci}
\address{Department of Mathematics, 
East Carolina University, 331 Austin Building, East Fifth St.,
Greenville, NC 27858-4353, USA}
\email{\mailto{fuccig@ecu.edu}}
\urladdr{\url{http://myweb.ecu.edu/fuccig/}}

\author[M. Piorkowski]{Mateusz Piorkowski}
\address{Department of Mathematics, Royal Institute of Technology (KTH) \\
Lindstedtsv\"agen 25, 100 44 Stockholm, Sweden}
\email{\href{mailto:mathpiorkowski@gmail.com}{mathpiorkowski@gmail.com}}
\urladdr{\url{https://sites.google.com/view/mateuszpiorkowski/home?pli=1}}

\author[J.\ Stanfill]{Jonathan Stanfill}
\address{Department of Mathematics, The Ohio State University \\
100 Math Tower, 231 West 18th Avenue, Columbus, OH 43210, USA}
\email{\mailto{stanfill.13@osu.edu}}
\urladdr{\url{https://u.osu.edu/stanfill-13/}}


\date{\today}
\subjclass[2020]{Primary: 47A10, 47B10, 47G10. Secondary: 34B27, 34L40.}
\keywords{$\zeta$-function, (quasi-regular) Sturm--Liouville operators, traces, zeta regularized functional determinants, (generalized) Bessel operators, Legendre differential equation.}

\begin{abstract}

In this work we analyze the spectral $\zeta$-function associated with the self-adjoint extensions, $T_{A,B}$, of quasi-regular Sturm--Liouville operators that are bounded from below. By utilizing the Green's function formalism, we find the characteristic function which implicitly provides the eigenvalues associated with a given self-adjoint extension $T_{A,B}$. The characteristic function is then employed to construct a contour integral representation for the spectral $\zeta$-function of $T_{A,B}$. By assuming a general form for the asymptotic expansion of the characteristic function, we describe the analytic continuation of the $\zeta$-function to a larger region of the complex plane. We also present a method for computing the value of the spectral $\zeta$-function of $T_{A,B}$ at all positive integers. 
We provide two examples to illustrate the methods developed in the paper: the generalized Bessel and Legendre operators.
We show that in the case of the generalized Bessel operator, the spectral $\zeta$-function develops a branch point at the origin, while 
in the case of the Legendre operator it presents, more remarkably, branch points at every nonpositive integer value of $s$.

\end{abstract}

\maketitle




\section{Introduction} 

The spectral $\zeta$-function is an invaluable tool both in the field of mathematics and mathematical physics. For instance, important information regarding a physical system in the ambit of quantum field theory can be extracted from the spectral $\zeta$-function of a suitable operator in a region to the left of its abscissa of convergence \cite{Em12,Em94,Ki02}. Additionally, the functional determinant of a self-adjoint operator representing the Hamiltonian $H$ of a quantum field is used to evaluate its one-loop effective action (see e.g. \cite{RS71,Toms07}). Furthermore, the analysis of the spectral $\zeta$-function at the point $s=-1/2$ is utilized to study the Casimir energy of a quantum field. In this setting, different self-adjoint extensions of the Hamiltonian describe the dynamics of the quantum field constrained to satisfy different types of boundary conditions (see e.g. \cite{BVW88,BKMM09}). Lastly, the residues at the simple poles of the spectral $\zeta$-function
of a given self-adjoint Hamiltonian $H$ represent the coefficients of the small-$t$ asymptotic expansion of the trace of the heat kernel associated to $H$. These coefficients are of particular interest because they describe the ultraviolet divergences that are present in a quantum system. Once these divergences are identified, they can be used to regularize the one-loop effective action \cite{BD82,Dwk84,DK78,Vas03}.
These are some of the most important applications of the spectral $\zeta$-function in the ambit of physics and more have been described in detail in \cite{Em12,Ki02}
 
The aim of this work is to extend the analysis of the spectral $\zeta$-function, which was developed for certain regular Sturm--Liouville operators under additional smoothness assumptions in \cite{FGKS21, FGK}, to quasi-regular Sturm--Liouville operators that are bounded from below. It is important to point out that while the spectral $\zeta$-function has been considered in the literature for specific singular operators (see e.g. \cite{FLP,Ki06,Ki08}), we provide here a systematic study of the $\zeta$-function of general quasi-regular Sturm--Liouville operators that are bounded from below, illustrating what is true in general and when additional assumptions are needed. 
The analysis of the spectral $\zeta$-function in the general quasi-regular setting is quite interesting since its structure, 
as a function in the complex plane, is expected to differ from the one obtained in the $N$-smooth regular case (see Definition \ref{defNsmooth}). 

To elaborate on this point, we briefly describe the process of analytic continuation of the spectral $\zeta$-function for the $N$-smooth regular case (details of which can be found in \cite{FGKS21, FGK}). The spectral $\zeta$-function associated with a given self-adjoint extension of a regular Sturm--Liouville operator is
represented as a contour integral of an expression containing a function whose zeroes implicitly provide the eigenvalues (the characteristic function). By construction, this representation is valid only on a finite strip of the complex plane. In order to extend the $\zeta$-function
to the left of this strip of convergence one subtracts, and then adds, to the integrand of the representation a number of terms 
of the asymptotic expansion of the characteristic function for large values of the independent variable. 
As a result of this process, the more terms of the asymptotic expansion that are subtracted and then added, the more the abscissa of convergence moves to the left. The specific form of the asymptotic expansion plays a crucial role in this extension process since it develops the structure of the 
$\zeta$-function. In the $N$-smooth regular case, the asymptotic expansion of the characteristic function contains an exponential, as a leading term, and 
inverse powers of the independent variable up to order depending on $N$. When subtracted and then added to the integrand, these terms give rise to simple poles in the spectral $\zeta$-function. (We show how this procedure can be extended to certain quasi-regular problems in Section \ref{cont}.) It is reasonable to expect, then, that any change to the form of the asymptotic expansion of the characteristic function
translates to a different structure of the spectral $\zeta$-function compared to the {\it standard} one obtained in the $N$-smooth regular case.
This is indeed what occurs in the general case and gives rise to a spectral $\zeta$-function which could have poles of higher order or branch points (as we will illustrate in Section \ref{examples} with an example having branch points at every nonnegative integer). It is interesting to point out that while it might appear that the nonintegrable nature of the endpoints can cause the spectral $\z$-function to develop singularities other than simple poles, according to the regularization given in Theorem \ref{regularizing} every quasi-regular problem can be transformed to an equivalent regular one. This means that the unusual behavior of the spectral $\z$-function can in fact occur for regular Sturm--Liouville problems as well! 
Spectral $\zeta$-functions presenting unusual properties have been studied very little in the literature and one of the objectives of this paper is to provide a class of problems for which such non-standard $\zeta$-functions are very likely to occur.
\medskip

The outline of this work is as follows. In the next section, we identify the self-adjoint extensions of quasi-regular Sturm--Liouville operators in terms of generalized boundary conditions. In Section \ref{zetafunction}, we provide an integral representation of the spectral $\zeta$-function
in terms of the trace of the resolvent of a given self-adjoint extension. The explicit expression for the trace of the resolvent is obtained 
by utilizing the Green's function formalism. This approach is different than the one provided in \cite{FGKS21,GK19}, which was based on the Fredholm determinant, and should be appealing to a wider audience. In Section \ref{cont}, we describe the method of analytic continuation 
of the spectral $\zeta$-function, while Section \ref{examples} illustrates our results through two examples: the generalized Bessel operator on 
$(0,b)\subset\R$ and the Legendre operator on $(-1,1)$. The last section contains some final remarks and Appendix \ref{appendix} provides the reader with 
a very brief summary of the basic notions of singular Weyl--Titchmarsh--Kodaira theory needed.   

We summarize, here, some of the notation used in this manuscript. If $A$ is a linear operator mapping (a subspace of) a Hilbert space into another, then $\dom(A)$ denotes the domain of $A$. The spectrum, point spectrum, and resolvent set of a closed linear operator in a separable complex Hilbert space, $\cH$, will be denoted by $\sigma(\dott),\ \sigma_p(\dott),$ and $\rho(\dott)$ respectively. If $S$ is self-adjoint in $\cH$, the multiplicity of an eigenvalue $z_0\in\sigma_p(S)$ is denoted $m(z_0;S)$ (the geometric and algebraic multiplicities of $S$ coincide in this case).
For consistency of notation, throughout this manuscript we will follow the conventional notion that derivatives annotated with superscripts are understood as with respect to $x$ and derivatives with respect to $z$ will be abbreviated by 
$\dot\ =d/d z$. We also employ the notation $\N_{0}=\N\cup\{0\}$.

\section{Self-adjoint quasi-regular Sturm--Liouville operators} \label{s2}
We begin our analysis by stating the following basic assumptions of this work:  
\begin{hypothesis} \label{h1}
Let $(a,b) \subseteq \bbR$ and suppose that $p,q,r$ are $($Lebesgue\,$)$ measurable functions on $(a,b)$ 
such that the following items $(i)$--$(iii)$ hold: \\[1mm] 
$(i)$ \hspace*{1.1mm} $r>0$ a.e.~on $(a,b)$, $r\in\Ll$. \\[1mm] 
$(ii)$ \hspace*{.1mm} $p>0$ a.e.~on $(a,b)$, $1/p \in\Ll$. \\[1mm] 
$(iii)$ $q$ is real-valued a.e.~on $(a,b)$, $q\in\Ll$. 
\end{hypothesis}
These assumptions represent the standard foundations over which the theory of singular Sturm--Liouville operators is developed (see e.g. \cite[Ch.~9]{Ze05}).
The functions introduced in the Hypothesis \ref{h1} are used to construct the following, general three-coefficient differential expression $\tau$:
\begin{align}\lb{2.1}
\tau=\f{1}{r(x)}\left[-\f{d}{dx}p(x)\f{d}{dx} + q(x)\right] \, \text{ for a.e.~$x\in(a,b) \subseteq \R$.} 
\end{align} 
A Sturm--Liouville operator is defined as the differential expression \eqref{2.1} together with a suitable domain in $\Lr$ over which $\tau$ acts.
Within the framework of self-adjoint extensions, the minimal and maximal operators play an important role.
By denoting with 
\begin{equation}
y^{[1]}(x) = p(x) y'(x), \quad x \in (a,b),
\end{equation}
the first quasi-derivative of a function $y\in AC_{loc}((a,b))$, the maximal and minimal operators are defined as follows.
\begin{definition} \lb{def1}
Assume Hypothesis \ref{h1}. Given $\tau$ as in \eqref{2.1}, then the \textit{maximal operator} $T_{max}$ in $\Lr$ associated with $\tau$ is defined by
\begin{align}
&T_{max} f = \tau f,    
\\
& f \in \dom(T_{max})=\big\{g\in\Lr \, \big| \,g,g^{[1]}\in\ACl;    
\tau g\in\Lr\big\},  \notag
\end{align}
The \textit{preminimal operator} $T_{min,0} $ in $\Lr$ associated with $\tau$ is defined by 
\begin{align}
&T_{min,0}  f = \tau f,   \notag
\\
&f \in \dom (T_{min,0})=\big\{g\in\Lr \, \big| \, g,g^{[1]}\in\ACl;   
\\
&\hspace*{3.25cm} \supp \, (g)\subset(a,b) \text{ is compact; } \tau g\in\Lr\big\}.   \notag
\end{align}
One can prove that $T_{min,0} $ is closable and we simply define the minimal operator, $T_{min}$, 
to be the closure of $T_{min,0}$.
\end{definition}
Since $(T_{min,0})^* = T_{max}$ \cite{Ze05}, one can conclude that $T_{max}$ is closed. This observation allows us to write that   
$T_{min}=\ol{T_{min,0} }$ is given by
\begin{align}
&T_{min} f = \tau f, 
\\
&f \in \dom(T_{min})=\big\{g\in\dom(T_{max})  \, \big| \, W(h,g)(a)=0=W(h,g)(b) \, 
\text{for all } h\in\dom(T_{max}) \big\},   \no 
\end{align}
where the Wronskian, $W(f,g)(x)$, of $f$ and $g$, for $f,g\in\ACl$, is defined by
\begin{equation}
W(f,g)(x) = f(x)g^{[1]}(x) - f^{[1]}(x)g(x), \quad x \in (a,b). 
\end{equation}
Since we are working in the singular setting, the self-adjoint extensions of $T_{min}$ depend on the limit point or limit circle 
classification of the endpoints of the interval $(a,b)$. In this paper we focus on the quasi-regular case which constrains the differential expression $\tau$ to be limit circle at the endpoints (see Definition \eqref{LCLP}). Since the most relevant systems that one encounters in physical applications are described by Hamiltonians that are bounded from below, in addition to the typical integrability hypothesis we will, for now on, work under the following assumption:
\begin{hypothesis}\label{h2}
The Sturm--Liouville differential expression $\tau$ satisfies Hypothesis \ref{h1}, is quasi-regular on $(a,b)$, and its minimal operator $T_{min}$ is bounded from below. 
\end{hypothesis}
We would like to point out that the addition of a suitable {\it mass} parameter $m>\lambda_{0}$ to $T_{min}$ (which represents a redefinition of the lowest energy state) leads to a positive operator, 
$(u,T_{min} u)_{L^2((a,b);rdx)}> 0,\; u \in \dom(T_{min})$. 

According to the theory of singular Sturm--Liouville operators, briefly outlined in Appendix \ref{appendix}, the self-adjoint extensions of $T_{min}$ are determined by Theorem \ref{extensions} with the generalized boundary values provided by Theorem \ref{At3}. We can summarize these results as follows:
\begin{theorem} \label{t}  
Assume Hypothesis \ref{h2}. Then every self-adjoint extension of $T_{min}$ belongs to either of the following two classes: \\
$(i)$ Separated boundary conditions are the self-adjoint extensions $T_{\al,\be}$ of $T_{min}$ of the form
\begin{align}
& T_{\al,\be} f = \tau f, \quad \al,\be\in[0,\pi),   \notag\\
& f \in \dom(T_{\al,\be})=\big\{g\in\dom(T_{max}) \, \big| \, \wti g(a)\cos(\al)+ {\wti g}^{\, \prime}(a)\sin(\al)=0;   \label{2.12text} \\ 
& \hspace*{5.5cm} \, \wti g(b)\cos(\be)- {\wti g}^{\, \prime}(b)\sin(\be) = 0 \big\}.    \notag
\end{align}
$(ii)$ Coupled boundary conditions are the self-adjoint extensions $T_{\varphi,R}$ of $T_{min}$ of the form
\begin{align}
\begin{split}\label{teq0} 
& T_{\varphi,R} f = \tau f,    \\
& f \in \dom(T_{\varphi,R})=\bigg\{g\in\dom(T_{max}) \, \bigg| \begin{pmatrix} \wti g(b)\\ {\wti g}^{\, \prime}(b)\end{pmatrix} 
= e^{i\varphi}R \begin{pmatrix}
\wti g(a)\\ {\wti g}^{\, \prime}(a)\end{pmatrix} \bigg\},
\end{split}
\end{align}
where $\varphi\in[0,\pi)$, and $R \in SL(2,\bbR)$ .  
\end{theorem}
The analysis of the spectral $\zeta$-function associated with a quasi-regular Sturm--Liouville operator represents a natural first step towards the treatment of 
spectral functions for more general singular Sturm--Liouville problems. This is due to the fact that quasi-regular problems can be 
{\it regularized} by using regularizing functions as detailed in \cite{NZ92} and \cite[Ch.~8]{Ze05}. In particular, many results that can be proved in 
the regular setting can be extended to the quasi-regular case as well. 
The following theorem combines the results proved in \cite[Thms. 8.2.1, 8.3.1, 8.3.2, 10.6.5, and Rem. 10.6.2]{Ze05} and represents an 
important result that details the link between regular and quasi-regular problems.

\begin{theorem} \label{regularizing}  
Assume Hypotheses \ref{h2}. Let $\hatt u(\lambda_0,\dott)\in\dom(T_{max})$ be positive on $(a,b)$ and a nonprincipal solution at $x=a,b$, but not necessarily a solution through the interior of $(a,b)$. Define
\begin{align}\label{regular0}
\begin{split}
P(x)=[\hatt u(\lambda_0,x)]^2 p(x),\quad R(x)=[\hatt u(\lambda_0,x)]^2 r(x),\quad Q(x)=r(x) \hatt u(\lambda_0,x) (\tau \hatt u)(\lambda_0,x).
\end{split}
\end{align}
Then the quasi-regular problem $\tau y(z,x)=zy(z,x)$ can be transformed into 
\begin{align}\label{regular1}
-(P(x)v'(z,x))'+Q(x)v(z,x)=z R(x) v(z,x),
\end{align}  
for which $1/P, Q, R\in L^1((a,b);dx)$.
\end{theorem}

This theorem simply states that by using a positive $\hatt u(\lambda_0,\dott)\in\dom(T_{max})$ which is nonprincipal at both endpoints, also named {\it regularizing function} \cite[Sect.~8.2]{Ze05}, one can 
transform a quasi-regular problem $\tau y(z,x)=zy(z,x)$ into an 
equivalent one \eqref{regular1} which is regular on $(a,b)$. In addition, one can prove the following auxiliary result \cite[Ch.~8]{Ze05}:
\begin{lemma}\label{lemma1}
Under the hypotheses of Theorem \ref{regularizing}, if $y(z,\dott)$ is a solution of the quasi-regular problem $\tau y(z,x)=zy(z,x)$ on $(a,b)$, then $v(z,\dott)=y(z,\dott)/\hatt u(\lambda_0,\dott)$
is a solution of the associated regular problem \eqref{regular1} on $(a,b)$. Similarly, if $v(z,\dott)$ is a solution of the regular problem \eqref{regular1} on $(a,b)$ then
$y(z,\dott)=\hatt u(\lambda_0,\dott)v(z,\dott)$
is a  solution of the quasi-regular problem $\tau y(z,x)=zy(z,x)$ on $(a,b)$. 
\end{lemma}
The last lemma is particularly useful for expressing, in a very simple way, the generalized boundary values of a function $g\in\dom(T_{max})$ when
$T_{min}$ is bounded from below (that is when the Hypothesis \ref{h1} holds). In fact, the following lemma provides the explicit formulas:
\begin{lemma}\label{lemma2}
Assume Hypothesis \ref{h2}. Let $y(z,\dott)$ be a solution of the quasi-regular problem $\tau y(z,x)=zy(z,x)$ on $(a,b)$ and let $\hatt u(\lambda_0,\dott)$
be a regularizing function on $(a,b)$. Then,
\begin{equation}
\wti y(z,x_0)=v(z,x_0),\quad
\wti y^{\, \prime}(z,x_0)=P(x_0) v'(z,x_0),\quad x_0\in\{a,b\}.
\end{equation}
\end{lemma}
\begin{proof}
Let $y(z,\dott)$ be a solution of the quasi-regular problem $\tau y(z,x)=zy(z,x)$ on $(a,b)$. Then one can write $y(z,\dott)=\hatt u(\lambda_0,\dott) v(z,\dott)$ where $v(z,\dott)$ is a solution of the associated regular problem. Choosing 
\begin{align}\label{nonprinc}
\hatt u(\lambda_0,x)=\begin{cases} \hatt u_a(\lambda_0,x), & \text{for $x$ near a}, \\
\hatt u_b(\lambda_0,x), & \text{for $x$ near b},  \end{cases}
\end{align}
(with $\hatt u(\lambda_0,x_0)=1$ near a regular endpoint $x=x_0$, i.e., in a neighborhood of $x=x_0$), we have $\wti y(z,x_0)=v(z,x_0)$ and $\wti y^{\, \prime}(z,x_0)=P(x_0) v'(z,x_0)$ for $x_0\in\{a,b\}$. The first expression is clear from \eqref{A6} and the second follows from \eqref{A7}. In fact, by denoting with $u(\lambda_0,x)$ the corresponding principal solution, one has
\begin{align}
\begin{split}
\wti y^{\, \prime}(z,x_0)&=\lim_{x\to x_0}\dfrac{\hatt u(\lambda_0,x) v(z,x)-\hatt u(\lambda_0,x) v(z,x_0)}{u(\lambda_0,x)}=\lim_{x\to x_0}\dfrac{v(z,x)-v(z,x_0)}{u(\lambda_0,x)/\hatt u(\lambda_0,x) } \\
&=\lim_{x\to x_0}\dfrac{v'(z,x)}{[\hatt u(\lambda_0,x)u'(\lambda_0,x)-\hatt u'(\lambda_0,x) u(\lambda_0,x) ]/[\hatt u(\lambda_0,x)]^2 } \\
&=\lim_{x\to x_0}\dfrac{[\hatt u(\lambda_0,x)]^2p(x) v'(z,x)}{p(x)[\hatt u(\lambda_0,x)u'(\lambda_0,x)-\hatt u'(\lambda_0,x) u(\lambda_0,x) ]}=P(x_0) v'(z,x_0), \quad x_0\in\{a,b\},
\end{split}
\end{align}
where we have used L'H\^opital's rule and the fact that $W(\hatt u(\lambda_0,\dott), u(\lambda_0,\dott)) = 1$.
\end{proof}

\section{The spectral  \texorpdfstring{$\zeta$}{zeta}-function}\label{zetafunction}

With the self-adjoint extensions now completely determined, we can move on to the analysis of the corresponding spectral $\zeta$-function.
Let us denote the separated and coupled self-adjoint extensions of $T_{min}$ in Theorem \ref{t} collectively by $T_{A,B}$.
The existence of the spectral $\zeta$-function for quasi-regular Sturm--Liouville operators is predicated on the following result \cite{GNZ23,Ze05}:
\begin{theorem} \label{t1}  
Assume Hypothesis \ref{h2} and let $T_{A,B}$ be a self-adjoint extension of $T_{min}$ with $z\in\rho(T_{A,B})$. Then, denoting by $\cB_1(\Lr)$ the space of trace class operators in $\Lr$,
\begin{equation}
(T_{A,B}-zI)^{-1}\in \cB_{1}(\Lr),
\end{equation}
and, hence, $T_{A,B}$ has a purely discrete spectrum with eigenvalues of multiplicity at most $2$. In addition, if $\sigma(T_{A,B})=\{\lambda_{A,B,j}\}_{j\in J}$
with $J\subset \Z$ an appropriate index set where eigenvalues are counted according to their multiplicity, then
\begin{equation}
\sum_{\underset{\lambda_j\neq 0}{j\in J}} |\lambda_{A,B,j}|^{-1}<\infty.
\end{equation}
\end{theorem}

This result allows us to define the spectral $\zeta$-function associated with $T_{A,B}$ as   
\begin{align}\lb{2.65}
\zeta(s;T_{A,B}):=\sum_{\underset{\lambda_j\neq 0}{j\in J}} \lambda_{A,B,j}^{-s},
\end{align}
for $\Re(s)>s_{0}>0$, with $s_0$ large enough. The goal is to extend the function $\zeta(s;T_{A,B})$ to a region of the complex plane to the left of $\Re(s)=s_0$. This is a necessary process to perform
since very often valuable information can be extracted from the spectral $\zeta$-function in the extended region $\Re(s)<s_0$ (see e.g. \cite{Ki02}). 
Throughout the relevant literature, the standard procedure used to extend the spectral $\zeta$-function to a larger region of the complex $s$-plane relies on 
a contour integral representation of $\zeta(s;T_{A,B})$ in \eqref{2.65} (this is exactly what was done recently for the regular case in \cite{FGKS21}). 
Let $\gamma$ be a counterclockwise contour in the complex plane that encircles the spectrum $\sigma(T_{A,B})$ and avoids the origin. Assume that the meromorphic function $[\textrm{tr}_{L_r^{2}(a,b)}(T_{A,B}-zI)^{-1}-z^{-1}m_{0}]$ (we write $L_r^{2}(a,b)$ instead of $\Lr$ in the subscript for brevity), with $m_{0}=m(0;T_{A,B})$, is polynomially bounded over $\gamma$. By \cite[Lemma 2.6]{GK19} one has, for $\Re(s)>s_0$,
\begin{equation}\label{1}
\zeta(s;T_{A,B})=-\frac{1}{2\pi i}\int_{\gamma}dz\,z^{-s}[\textrm{tr}_{L_r^{2}(a,b)}(T_{A,B}-zI)^{-1}+z^{-1}m_{0}].
\end{equation}    
Due to the presence of $z^{-s}$ we choose, following \cite{KM03} (see also \cite{KM04}), the branch cut for the definition of the integral to be 
\begin{align}\label{1a}
    R_\psi=\{z=te^{i\psi}:t\in [0,\infty)\},\quad \psi\in (\pi/2,\pi).
\end{align}
 
Before we can proceed with the process of analytic continuation, it is necessary to find a more explicit expression for the trace of the resolvent 
of $T_{A,B}$. This can be achieved with the use of the Green's function. For $z\in\rho(T_{A,B})$ the resolvent $(T_{A,B}-zI)^{-1}$ can be expressed as \cite{DS88,GNZ23}
\begin{equation}
(T_{A,B}-zI)^{-1}f(x)=\int_{a}^{b}r(y)dy\,G_{A,B}(z,x,y)f(y),
\end{equation}
for $f\in\Lr$. The Green's function $G_{A,B}(z,x,y)$ solves the equation 
\begin{equation}\label{2}
(\tau-z)G_{A,B}(z,x,y)=\delta(x-y),
\end{equation} 
on $(a,b)$ endowed with the conditions
\begin{equation}\label{3}
G_{A,B}(z,x,y)|_{x=y^{+}}=G_{A,B}(z,x,y)|_{x=y^{-}},\quad
\frac{\partial G_{A,B}}{\partial x}(z,x,y)|_{x=y^{-}}-\frac{\partial G_{A,B}}{\partial x}(z,x,y)|_{x=y^{+}}=p^{-1}(y),
\end{equation}
and the boundary conditions that characterize the specific self-adjoint extension $T_{A,B}$ \cite{We87}. In order to compute the 
Green's function, we introduce the fundamental system of solutions $\theta(z,x,a)$, $\phi(z,x,a)$ of $\tau y(z,x)=z y(z,x)$ defined by
\begin{align}\label{4}
\wti\theta(z,a,a)=\wti\phi^{\, \prime}(z,a,a)=1,\quad \wti\theta^{\, \prime}(z,a,a)=\wti\phi(z,a,a)=0,
\end{align}
such that
\begin{align}
    W\big(\theta(z,\dott,a),\phi(z,\dott,a)\big)=1.
\end{align}
In terms of the newly introduced $\theta(z,x,a)$, $\phi(z,x,a)$, the Green's function can be written as 
\begin{equation}\label{5}
G_{A,B}(z,x,y) =
\begin{cases}
a_{A,B}(y)\phi(z,x,a)+b_{A,B}(y) \theta(z,x,a),& a<x<y<R, \\
c_{A,B}(y)\phi(z,x,a)+d_{A,B}(y) \theta(z,x,a),& a<y<x<R,
\end{cases}
\end{equation}
where the unknown terms $(a_{A,B}(y),b_{A,B}(y),c_{A,B}(y),d_{A,B}(y))$ are determined by imposing the conditions \eqref{3} and either
the boundary conditions \eqref{2.12text} for the separated case or the boundary conditions \eqref{teq0} for the coupled case.

By introducing the generalized boundary value operators for $f(z,\dott),f^{[1]}(z,\dott)\in AC_{loc}((a,b))$,
\begin{equation}\label{BCop}
U_{a}(f)(z)=\cos(\alpha)\wti f(z,a)+\sin(\alpha)\wti f^{\,\prime}(z,a),\quad
U_{b}(f)(z)=\cos(\beta)\wti f(z,b)-\sin(\beta)\wti f^{\,\prime}(z,b),
\end{equation}
for the self-adjoint extensions $T_{\alpha,\beta}$ belonging to the separated case, the Green's function in \eqref{5} can be found to be
\begin{equation}\label{6}
G_{\alpha,\beta}(z,x,y)=\frac{1}{F_{\alpha,\beta}(z)}\Big[\cos(\alpha)\phi(z,x,a)-\sin(\alpha)\theta(z,x,a)\Big]\Big[U_{b}(\phi)(z)\theta(z,y,a)-U_{b}(\theta)(z)\phi(z,y,a)
\Big],
\end{equation}
for $a<x<y<R$ with 
\begin{equation}\label{7}
F_{\a,\b}(z)=\cos(\a)\,U_{b}(\phi)(z)-\sin(\a)\,U_{b}(\theta)(z).
\end{equation}
For $a<y<x<R$, the Green's function has the same expression as \eqref{6} but with $x\leftrightarrow y$. 

Furthermore, by introducing the generalized boundary value operators for $f(z,\dott),f^{[1]}(z,\dott)\in AC_{loc}((a,b))$,
\begin{align}\label{BCop1}
V_{1}(f)(z)&=\wti f(z,a)-e^{i\varphi}R_{22}\wti f(z,b)+e^{i\varphi}R_{12}\wti f^{\,\prime}(z,b),\nonumber\\
V_{2}(f)(z)&=\wti f^{\,\prime}(z,a)+e^{i\varphi}R_{21}\wti f(z,b)-e^{i\varphi}R_{11}\wti f^{\,\prime}(z,b),
\end{align}
the Green's function associated with the coupled self-adjoint extensions $T_{\varphi,R}$ has the form
\begin{align}\label{8}
G_{\varphi,R}(z,x,y)&=\frac{1}{F_{\varphi,R}(z)}\Big\{[\theta(z,y,a)V_{2}(\phi)(z)-\phi(z,y,a)V_{2}(\theta)(z)]\phi(z,x,a)\nonumber\\
&\quad+[\theta(z,y,a)V_{1}(\phi)(z)-\phi(z,y,a)V_{1}(\theta)(z)]\theta(z,x,a)\Big\},
\end{align}
for $a<x<y<R$ with 
\begin{equation}\label{9}
F_{\varphi,R}(z)=V_{1}(\theta)(z)+V_{2}(\phi)(z)+e^{2i\varphi}-1.
\end{equation}
Once again, the expression in the region $a<y<x<R$ is the same as \eqref{8} with the replacement $x\leftrightarrow y$. 
We would like to point out that for $z\in\rho(T_{\alpha,\beta})$ the denominators $F_{\a,\b}(z)\neq 0$ and $F_{\varphi,R}(z)\neq 0$ since there cannot exist any non-vanishing solution of $\tau g=zg$ satisfying the boundary conditions \eqref{2.12text}, for the separated case and \eqref{teq0} for the coupled case. 

Since the resolvent $(T_{A,B}-zI)^{-1}$ is of trace class, its trace can be computed according to the formula 
\begin{equation}\label{10}
\textrm{tr}_{L_r^{2}(a,b)}(T_{A,B}-zI)^{-1}=\int_{a}^{b}r(x)dx\,G_{A,B}(z,x,x).
\end{equation}
It is clear from the explicit expressions of $G_{\alpha,\beta}(z,x,y)$ in \eqref{6} and $G_{\varphi,R}(z,x,y)$ in \eqref{8}  
that the trace of the resolvent in \eqref{10} reduces to a combination of integrals of the products $\theta(z,x,a)\phi(z,x,a)$, $\theta^{2}(z,x,a)$, and $\phi^{2}(z,x,a)$. These integrals can be explicitly computed by using the following:
\begin{lemma}\label{lemma3}
Let $f_{1}(z,x)$ and $f_{2}(z,x)$ be solutions of the quasi-regular problem $\tau f_{i}(z,x)=z  f_{i}(z,x)$ with $z\in\mathbb{C}$ and $\tau$ given in \eqref{2.1} on the interval $(a,b)$. Then
\begin{equation}\label{lemma3e0}
\int_{a}^{b}r(x)dx \,f_{1}(z,x)f_{2}(z,x)=\wti W(f_{1},\dot{f}_{2})(a)-\wti W(f_{1},\dot{f}_{2})(b),
\end{equation} 
where
\begin{equation}
\wti W(f_{1},f_{2})(x_{0})=\wti f_{1}(z,x_0)\wti f^{\,\prime}_{2}(z,x_0)-\wti f^{\, \prime}_{1}(z,x_0)\wti f_{2}(z,x_0),\quad x_0\in\{a,b\}.
\end{equation}
\end{lemma}
\begin{proof}
According to Lemma \ref{lemma1} the functions $v_{i}(z,x)=f_{i}(z,x)/\hatt u(\lambda_0,x)$, with $i\in\{1,2\}$ and $\hatt u(\lambda_0,\dott)$ a regularizing function on $(a,b)$, are solutions of the corresponding regular problem \eqref{regular1}. This implies that 
\begin{equation}\label{lemma3e1}
\int_{a}^{b}r(x)dx\,f_{1}(z,x)f_{2}(z,x)=\int_{a}^{b}R(x)dx\,v_{1}(z,x)v_{2}(z,x).
\end{equation}
The integrand on the right-hand side can be expressed as a derivative 
\begin{equation}\label{lemma3e2}
R(x)v_{1}(x,z)v_{2}(x,z)=-\frac{d}{dx}W(v_{1},\dot{v}_{2})(x).
\end{equation}
This fact can be proved by a direct calculation by noticing that 
\begin{equation}
\frac{d}{dx}W(v_{1},\dot{v}_{2})=v_{1}(z,x)(P(x)\dot{v}^{\,\prime}_{2}(z,x))^{\,\prime}-(P(x)v^{\,\prime}_{1}(z,x))^{\,\prime}\dot{v}_{2}(z,x),
\end{equation}
and that $v_{1}(z,x)$ satisfies \eqref{regular1} and $\dot{v}_{2}(z,x)$ satisfies, instead, 
\begin{equation}
-(P(x)\dot{v}^{\,\prime}_{2})^{\,\prime}+Q(x)\dot{v}_{2}=R(x)v_{2}+z R(x)\dot{v}_{2}.
\end{equation}
From \eqref{lemma3e2} we then have
\begin{equation}
\int_{a}^{b}R(x)dx\,v_{1}(z,x)v_{2}(z,x)=W(v_{1},\dot{v}_{2})(a)-W(v_{1},\dot{v}_{2})(b).
\end{equation}
According to Lemma \ref{lemma2}, for $x_0\in\{a,b\}$ one has
\begin{align}
W(v_{1},\dot{v}_{2})(x_0)&=v_{1}(z,x_0)(P(x_0)\dot{v}^{\,\prime}_{2}(z,x_0))-(P(x_0)v^{\,\prime}_{1}(z,x_0))\dot{v}_{2}(z,x_0)\nonumber\\
&=\wti f_{1}(z,x_0)\wti{\dot{f^{\,\prime}_{2}}}(z,x_0)-\wti f_{1}^{\,\prime}(z,x_0)\wti{\dot{f_{2}}}(z,x_0)=\wti W(f_{1},\dot{f}_{2})(x_0),
\end{align}
which, together with \eqref{lemma3e1}, proves \eqref{lemma3e0}.
\end{proof}
\begin{remark}
Let us notice that the normalization of the fundamental system of solutions $\phi(z,\dott,a)$, $\theta(z,\dott,a)$ of $\tau y(z,x)=zy(z,x)$ given in \eqref{4} implies that $\phi(z,x,a)=\hatt u(\lambda_0,x) \varphi(z,x,a)$ and $\theta(z,x,a)=\hatt u(\lambda_0,x) \vartheta(z,x,a)$ where $\varphi(z,\dott,a),$ $\vartheta(z,\dott,a)$ is the fundamental system of solutions for the associated regular problem satisfying
\begin{align}
\vartheta(z,a,a)=\varphi^{[1]}(z,a,a)=1,\quad \vartheta^{[1]}(z,a,a)=\varphi(z,a,a)=0.
\end{align}
This observation and Lemma \ref{lemma2} allows us to conclude that 
\begin{equation}
\wti{\dot{\phi}}(z,a,a)=\wti{\dot{\theta}}(z,a,a)=\wti{\dot{\phi^{\,\prime}}}(z,a,a)=\wti {\dot{\theta^{\,\prime}}}(z,a,a)=0,
\end{equation}
which, in particular, means that 
\begin{equation}\label{11}
\wti W(\phi,\dot{\theta})(a)=\wti W(\theta,\dot{\phi})(a)=\wti W(\theta,\dot{\theta})(a)=\wti W(\phi,\dot{\phi})(a)=0.
\end{equation}
\hfill$\diamond$
\end{remark}
By using Lemma \ref{lemma3} and the result \eqref{11} it is not difficult to obtain for the trace in \eqref{10},
\begin{align}
\textrm{tr}_{L_r^{2}(a,b)}(T_{\alpha,\beta}-zI)^{-1}&=\frac{1}{F_{\alpha,\beta}(z)}\Big[\cos(\alpha)U_{b}(\theta)(z)\wti W(\phi,\dot{\phi})(b)+\sin(\alpha)U_{b}(\phi)(z)\wti W(\theta,\dot{\theta})(b)\nonumber\\
&\quad-\cos(\alpha)U_{b}(\phi)(z)\wti W(\theta,\dot{\phi})(b)-\sin(\alpha)U_{b}(\theta)(z)\wti W(\phi,\dot{\theta})(b)\Big].
\end{align}
Now, a straightforward calculation shows, 
\begin{align}
\begin{split}
&U_{b}(\phi)(z)\wti W(\theta,\dot{\phi})(b)-U_{b}(\theta)(z)\wti W(\phi,\dot{\phi})(b)=U_{b}(\dot{\phi})(z),\\
&U_{b}(\theta)(z)\wti W(\phi,\dot{\theta})(b)-U_{b}(\phi)(z)\wti W(\theta,\dot{\theta})(b)=U_{b}(\dot{\theta})(z),
\end{split}
\end{align}
which, finally, allows us to obtain
\begin{equation}\label{12}
\textrm{tr}_{L_r^{2}(a,b)}(T_{\alpha,\beta}-zI)^{-1}=-\frac{\cos(\alpha)U_{b}(\dot{\phi})(z)-\sin(\alpha)U_{b}(\dot{\theta})(z)}{F_{\alpha,\beta}(z)}=-\frac{d}{dz}\ln\,F_{\alpha,\beta}(z).
\end{equation}

A similar analysis can be performed for the trace of the resolvent for coupled self-adjoint extensions. By using \eqref{8}
in \eqref{10} and the results of Lemma \ref{lemma3} we find
\begin{align}\label{13}
\textrm{tr}_{L_r^{2}(a,b)}(T_{\varphi,R}-zI)^{-1}&=\frac{1}{F_{\varphi,R}(z)}\Big\{V_{2}(\theta)(z)\wti W(\phi,\dot{\phi})(b)
-V_{1}(\phi)(z)\wti W(\theta,\dot{\theta})(b)\nonumber\\
&\quad-V_{2}(\phi)(z)\wti W(\theta,\dot{\phi})(b)+V_{1}(\theta)(z)\wti W(\phi,\dot{\theta})(b)\Big\}.
\end{align}
The terms in the brackets can now be simplified to write
\begin{equation}\label{14}
\textrm{tr}_{L_r^{2}(a,b)}(T_{\varphi,R}-zI)^{-1}=-\frac{V_{1}(\dot{\theta})(z)+V_{2}(\dot{\phi})(z)}{F_{\varphi,R}(z)}=-\frac{d}{dz}\ln\,F_{\varphi,R}(z).
\end{equation}

The results obtained in \eqref{12} and \eqref{14} are summarized in the following:
\begin{theorem}\label{t2}
Assume that Hypothesis \ref{h2} holds. Denote by $T_{A,B}$ the self-adjoint extensions of $T_{min}$ described in Theorem \ref{t}. Then
\begin{equation}\label{15}
\textrm{tr}_{L_r^{2}(a,b)}(T_{A,B}-zI)^{-1}=-\frac{d}{dz}\ln\,F_{A,B}(z),
\end{equation} 
with the characteristic function $F_{A,B}(z)$ given by \eqref{7} for the separated case and by \eqref{9} for the coupled case.
\end{theorem}
We would like to point out that \eqref{7} and \eqref{9} show that $F_{A,B}(z)$ is an entire function (since so is the fundamental set of solutions) with isolated zeroes at the eigenvalues of $T_{A,B}$ (with multiplicity of zeros and eigenvalues agreeing). This implies, from \eqref{15}, that the trace of the resolvent of the self-adjoint operator $T_{A,B}$ is a meromorphic function of $z\in\bbC$ with 
isolated simple poles at the real eigenvalues $\lambda_{A,B,j}\in\sigma(T_{A,B})$.

The explicit expression for the trace of the resolvent obtained in \eqref{15} allows us to rewrite the integral representation of the spectral $\zeta$-function in \eqref{1} as
\begin{equation}\label{16}
\zeta(s;T_{A,B})=\frac{1}{2\pi i}\int_{\gamma}dz\,z^{-s}\left[\frac{d}{dz}\ln\,F_{A,B}(z)-z^{-1}m_{0}\right],
\end{equation} 
valid for $\Re(s)>s_0$. The process of analytic continuation of $\zeta(s;T_{A,B})$ begins with the deformation of the integration contour $\gamma$ to a new one that encloses
the branch cut $R_\psi$ as detailed in \cite{FGKS21,GK19}. As the contour shrinks to the branch cut $R_\psi$, we need to analyze and control the behavior of the integrand of \eqref{16} for $|z|\to\infty$ and for $|z|\to 0$. Since the zero eigenvalue (if it exists for a given self-adjoint extension) is excluded from the integral representation
\eqref{16}, the integrand has the behavior 
\begin{equation}\label{17}
\left|\frac{d}{dz}\ln\,F_{A,B}(z)-z^{-1}m_{0}\right|=C+O(|z|) ,\quad |z|\to 0,
\end{equation}  
with $C$ being a constant. As we have mentioned earlier, $F_{A,B}(z)$ is an entire function with zeroes coinciding with the eigenvalues $\lambda_{A,B,j}\in\sigma(T_{A,B})$.
According to the Weierstrass factorization theorem, the convergence of the spectral $\zeta$-function \eqref{2.65} for $\Re(s)>s_0$ implies that $F_{A,B}(z)$ is of 
finite order \cite[Sect.~13.15]{Nev07}. In fact, we have the following result in case the eigenvalues grow polynomially. 
    
\begin{proposition}\label{Prop:lnF}
        Consider a self-adjoint extension $T_{A,B}$ of the bounded below minimal operator $T_{min}$ with a purely discrete spectrum $\sigma(T_{A,B}) = \lbrace \lambda_n \rbrace_{n = 1}^\infty$. Assume additionally the eigenvalue asymptotics $\lambda_n = c n^\kappa[1+o(1)]$ with some $\kappa > 1$ as $n \to \infty$. Then the characteristic function, $F_{A,B}(z)$, satisfies
        \begin{align}\label{lnF}
            -\frac{d}{dz} \ln \, F_{A,B}(z) = \frac{2\pi i}{(1-e^{-2 \pi i/\kappa})\kappa c^{1/\kappa}} z^{\frac{1}{\kappa}-1} + o(z^{\frac{1}{\kappa}-1}), \quad z \to \infty,
        \end{align}
        uniformly in closed sectors of $\C$ not containing the positive real line. Here, the $\kappa^{th}$-root is chosen to have a branch cut on the positive real axis with positive limit from the upper-half plane.
    \end{proposition}
    Note that with a polynomial eigenvalue growth as above, the infinite sum representation of the spectral $\zeta$-function will converge for all $s$ with $\Re(s) > \kappa^{-1}$.
    \begin{proof}[Proof of Proposition \ref{Prop:lnF}]
        Let us introduce the spectral measures $d\mu = \sum_{n = 1}^\infty \delta_{\lambda_n}$ with unit point masses where the eigenvalues are counted according to multiplicity. Then we have the equality
        \begin{align}
            -\frac{d}{dz} \ln \, F_{A,B}(z) = \int_{\R} \frac{1}{x-z} d\mu(x), \quad z \in \C \setminus \sigma(T_{A,B}).
        \end{align}
        Now for $t > 0$ introduce the scaled variables $w = z/t$, $y = x/t$ and set $f_t(w) = -t^{1-1/\kappa}\frac{d}{dz} \ln \, F_{A,B}(z)|_{z=wt}$.
        Then for $f_t$ we have the integral representation
        \begin{align}
            f_t(w) =  \int_{\R} \frac{1}{y-w} d\mu_t(y),
        \end{align}
        with $d\mu_t = t^{-1/\kappa} \sum_{n=1}^\infty \delta_{\lambda_n/t}$. Now as $\lambda_n \sim cn^\kappa$, we have that
        \begin{align}
            \int_{-\infty}^Y d\mu_t(y) = t^{-1/\kappa} \cdot \# \lbrace \lambda_n \colon \lambda_n \leq tY \rbrace \sim 
            \begin{cases}
               \big( \frac{Y}{c}\big)^{1/\kappa} = \int_{0}^Y \frac{1}{\kappa c^{1/\kappa}y^{1-1/\kappa}} dy,  & \text{for } Y \geq 0, 
               \\
               0, & \text{otherwise}.
            \end{cases}
        \end{align}
        This shows that $d \mu_t(y) \to \chi_{[0, \infty)}\frac{1}{\kappa c^{1/\kappa}y^{1-1/\kappa}} dy$ in distribution. In particular, we have that pointwise
        \begin{align}
            f_t(w) \to \int_0^\infty \frac{1}{y-w} \frac{1}{\kappa c^{1/\kappa}y^{1-1/\kappa}} dy, \quad t \to \infty, \ w \in \C \setminus [0, \infty).
        \end{align}
        Now a simple residue calculation yields 
        \begin{align*}
        f_t(w) \to \frac{2\pi i}{(1-e^{2 \pi i/\kappa})\kappa c^{1/\kappa}w^{1-1/\kappa}},
        \end{align*}
        where the $\kappa^{th}$-root is chosen to have a branch cut on the positive real axis with positive limit from the upper-half plane. The convergence rate is locally uniform in $\C \setminus [0, \infty)$.

        After reversing our change of variables we obtain \eqref{lnF} where the convergence is uniform in closed sectors not containing the real line. This finishes the proof.
    \end{proof}
    In general, results allowing us to determine from the growth properties of a measure $\mu$ the asymptotic behavior of its Stieltjes transform at infinity are known in the literature as \emph{Abelian theorems} and have been studied in great detail (see \cite{LW24} for some recent developments).
    
We would like to point out that by integrating and then exponentiating \eqref{lnF} results in the asymptotics 
        \begin{align}
            F_{A,B}(z) = \exp\Big(-\frac{2\pi i}{(1-e^{2 \pi i/\kappa})c^{1/\kappa}}z^{1/\kappa} + o(z^{1/\kappa})\Big),
        \end{align}
again with the same uniform convergence rate as before, consistent with the fact that $F_{A,B}(z)$ is of order $\rho = \kappa^{-1}$. 

According to the result \eqref{17} and the proposition above one concludes that, in the process of shrinking the integration contour to the branch cut $R_\psi$, the contributions to the integral as $|z|\to 0$ and $|z|\to\infty$ vanish as long as $\rho<\Re(s)<1$. That is, this first step in the analytic continuation of the spectral $\z$-function is well-defined
only when the characteristic function $F_{A,B}(z)$ has order $\rho<1$. This is, however, guaranteed, as proved in Proposition \ref{Prop:lnF}, for problems with eigenvalues satisfying $\lambda_n = c n^\kappa[1+o(1)]$ with some $\kappa > 1$ as $n \to \infty$.

In particular, for the quasi-regular case we know \cite[Thm. 10.6.1]{Ze05} that the eigenvalues of any self-adjoint extension satisfy the Weyl asymptotics
         \begin{align}\lb{WeylAsym}
             \lambda_n = \pi^2 n^2 \Bigg( \int_a^b \sqrt{\frac{r(x)}{p(x)}} \, dx \Bigg)^{-2}[1+o(1)]\quad \text{ as $n\to\infty$}, \quad \lambda_n \in \sigma(T),
         \end{align}
which imply, by applying Proposition \ref{Prop:lnF}, that the corresponding characteristic function satisfies \begin{align}\label{dfasymp}
            -\frac{d}{dz} \ln \, F_{A,B}(z) = \frac{i}{2} \Bigg( \int_a^b \sqrt{\frac{r(x)}{p(x)}} \, dx \Bigg) z^{-1/2} + o(z^{-1/2}), \quad z \to \infty,
        \end{align}
in closed sectors not containing the positive real line and the square root being defined by $z^{1/2} = e^{i\theta/2} r^{1/2}$ for $z = e^{i\theta} r$ with $r > 0$ and $\theta \in (0, 2\pi)$. It is then clear that, in the special case of quasi-regular problems, the characteristic function is always of order $1/2$.

Hence, by shrinking the integration contour $\gamma$ to the branch cut $R_\psi$ for $1/2<\Re(s)<1$, the contributions to the integral for $|z|\to 0$ and $|z|\to \infty$ vanish and one obtains
\begin{align}\label{19}
\zeta(s;T_{A,B})&=\frac{1}{2\pi i}\int_{\gamma}dz\,z^{-s}\left[\frac{d}{dz}\ln\,F_{A,B}(z)-z^{-1}m_{0}\right]\nonumber\\
&=\frac{1}{2\pi i}\int_{\gamma}dz\,z^{-s}\frac{d}{dz}\ln\,[z^{-m_0}F_{A,B}(z)]\nonumber\\
&=-\frac{e^{-is\Psi}}{2\pi i}\int_{0}^{\infty}dt\,t^{-s}\frac{d}{dt}\ln\left[(te^{i\Psi})^{-m_0}F_{A,B}\left(te^{i\Psi}\right)\right]\nonumber\\
&\quad+\frac{e^{-is(\Psi-2\pi)}}{2\pi i}\int_{0}^{\infty}dt\,t^{-s}\frac{d}{dt}\ln\left[(te^{i(\Psi-2\pi)})^{-m_0}F_{A,B}\left(te^{i(\Psi-2\pi)}\right)\right]\nonumber\\
&=e^{is(\pi-\Psi)}\frac{\sin(\pi s)}{\pi}\int_{0}^{\infty}dt\,t^{-s}\frac{d}{dt}\ln\left[(te^{i\Psi})^{-m_0}F_{A,B}\left(te^{i\Psi}\right)\right].
\end{align}
The calculation outlined in \eqref{19} serves as a proof of the following:
\begin{theorem}\label{t3}
Assume that Hypothesis \ref{h2} holds. Denote by $T_{A,B}$ the self-adjoint extensions of $T_{min}$ described in Theorem \ref{t}. Let $R_\psi$ be the branch cut defined in \eqref{1a}, $m_0$ be the multiplicity of the zero eigenvalue, and $F_{A,B}$ be the characteristic function \eqref{7} for the separated case or \eqref{9} for the coupled case. Then the spectral $\zeta$-function associated with $T_{A,B}$ has the integral representation 
\begin{equation}\label{19a}
\zeta(s;T_{A,B})=e^{is(\pi-\Psi)}\frac{\sin(\pi s)}{\pi}\int_{0}^{\infty}dt\,t^{-s}\frac{d}{dt}\ln\left[(te^{i\Psi})^{-m_0}F_{A,B}\left(te^{i\Psi}\right)\right],
\end{equation} 
which is valid for $1/2<\Re(s)<1$. 
\end{theorem}

\subsection{The spectral \texorpdfstring{$\zeta$}{zeta}-function at positive integers}
\hfill

As we have already mentioned in the introduction, information that is important in both physical and mathematical applications is 
found from the analysis of the spectral $\zeta$-function in a region to the left of the strip of convergence $1/2<\Re(s)<1$ established 
in Theorem \ref{t3}. Before describing in detail the process of analytic continuation that extends the spectral $\zeta$-function further 
to the left, we would like to present a simple way of evaluating $\zeta(s;T_{A,B})$ at positive integers. The values $\zeta(n;T_{A,B})$ for 
$n\in\bbN$ were first explicitly computed for the self-adjoint extensions $T_{A,B}$ of regular Sturm--Liouville
operators in \cite{FGKS21} and can be obtained directly from the integral representation \eqref{16}. The analysis of $\zeta(n;T_{A,B})$ is of particular interest because it provides an example in which the value of the spectral $\zeta$-function can be computed exactly. Moreover, whenever the spectral $\zeta$-can be expressed as a multiple of the classic $\zeta$-function of Riemann or Hurwitz, the values at positive integers can directly be used to find information to the left of $\Re(s)=1/2$.

More precisely, by letting $s=n$, $n\in\N$, in \eqref{16}, one no longer needs the branch cut $R_\psi$ for the fractional powers of $z^{-s}$, thus reducing the integral along the curve $\gamma$ to a clockwise oriented integral along the circle $C_\varepsilon$, centered at zero with 
radius $\varepsilon>0$. Considering the case $s=n$ also ensures that $m_0$ does not contribute to the integral in \eqref{16}. Hence,
\begin{align}\label{20}
\begin{split}
\zeta(n;T_{A,B})&=-\dfrac{1}{2\pi i}\ointctrclockwise_{C_\epsilon}dz\ z^{-n}\dfrac{d}{dz}\ln (F_{A,B}(z))=-\Res \left[ z^{-n}\dfrac{d}{dz}\ln (F_{A,B}(z));\ z=0 \right],\quad n\in\N.
\end{split}
\end{align} 
It is clear, from \eqref{20} that the value $\zeta(n;T_{A,B})$ is proportional to the coefficient of $z^{n-1}$ of the expansion of the 
characteristic function $F_{A,B}(z)$ in a neighborhood of $z=0$.  Since $F_{A,B}(z)$ is, in particular, an entire function of $z\in\C$, 
its small-$z$ expansion exists and can be obtained as follows:
from the explicit expressions provided in \eqref{7} for the separated 
case or \eqref{9} for the coupled case, the small-$z$ expansion of $F_{A,B}(z)$ is obtained from the one for the generalized boundary values 
$\wti \phi(z,x_0,a)$, $\wti\theta(z,x_0,a)$, $\wti\phi^{\,\prime}(z,x_0,a)$, and $\wti\theta^{\,\prime}(z,x_0,a)$, $x_0\in\{a,b\},$ of the fundamental 
system of solutions $\phi(z,\dott,a)$ and $\theta(z,\dott,a)$ of $\tau y(z,x)=zy(z,x)$.
According to Theorem \ref{regularizing} one has
\begin{align}\label{21}
\begin{split}
\wti \phi(z,x_0,a)&=\varphi(z,x_0,a),\quad\wti\theta(z,x_0,a)=\vartheta(z,x_0,a),\\ 
\wti\phi^{\,\prime}(z,x_0,a)&=\varphi^{[1]}(z,x_0,a),
\quad\wti\theta^{\,\prime}(z,x_0,a)=\vartheta^{[1]}(z,x_0,a),
\end{split}
\end{align} 
where $\varphi(z,\dott,a)$ and $\vartheta(z,\dott,a)$ represent the fundamental system of solutions for the associated regular problem.
Now, the small-$z$ expansion of the fundamental system of solutions $\varphi(z,\dott,a),$ $\vartheta(z,\dott,a)$ of the regular Sturm--Liouville problem and their quasi-derivative $\varphi^{[1]}(z,\dott,a),$ $\vartheta^{[1]}(z,\dott,a)$ has been explicitly obtained, by utilizing the Volterra integral expansion, in \cite[Subsection 3.1]{FGKS21}. This implies that by using the expansions found in \cite[Subsection 3.1]{FGKS21} on the right-hand side of the relations \eqref{20} we obtain the small-$z$ expansion of the generalized boundary values of the fundamental 
system of solutions $\phi(z,\dott,a)$ and $\theta(z,\dott,a)$ and, consequently, the one for $F_{A,B}(z)$. We can now state the following:
\begin{theorem}\label{t4}
Assume Hypothesis \ref{h2}. Denote by $T_{A,B}$ the self-adjoint extension of $T_{min}$ with either separated or coupled boundary conditions as described in Theorem \ref{t}. Denote by 
\begin{align}\label{22}
    F_{A,B}(z)=\sum_{j=0}^\infty a_jz^j,
\end{align}
the series expansion of the characteristic function $F_{A,B}(z)$, given in \eqref{7} for the separated 
case or \eqref{9} for the coupled case.
Then for $n\in\N$,
\begin{align}\label{23}
    \zeta(n;T_{A,B})=-\Res \left[ z^{-n}\dfrac{d}{dz}\ln (F_{A,B}(z));\ z=0 \right]=-n b_n,
\end{align}
where
\begin{equation}\label{24}
        b_1=\dfrac{a_{1+m_0}}{a_{m_0}},\quad
        b_j=\dfrac{a_{j+m_0}}{a_{m_0}}-\sum_{\ell=1}^{j-1}\left(\dfrac{\ell}{j}\right)\dfrac{a_{j-\ell+m_0}}{a_{m_0}}b_{\ell},\quad j\geq 2.
\end{equation}
In particular, if zero is not an eigenvalue of $T_{A,B}$, then
\begin{align}\label{25}
    \text{\rm tr}_{L^2_r((a,b))}\big(T_{A,B}^{-1}\big)=\zeta(1;T_{A,B})=-\dfrac{a_{1}}{a_{0}}.
\end{align}
\end{theorem}
 \begin{proof}
The proof is analogous to that of \cite[Thm. 4.1]{FGKS21}.
\end{proof}
We would like to point out that while the small-$z$ expansion of the characteristic function $F_{A,B}(z)$ of a given separated or
coupled self-adjoint extension can be obtained as described above through the small-$z$ expansion of the generalized boundary values of 
the fundamental system of solutions, this procedure appears to be a somewhat cumbersome exercise. In particular, because it involves computing the Volterra integral expansion of the corresponding regularized fundamental system of solutions.
Fortunately, for the majority of problems relevant in applications, it is not necessary to follow 
that procedure. In fact, very often, the characteristic function of a given self-adjoint extension is expressed 
in terms of well-known special functions and its small-$z$ expansion can be computed from those of its constituent 
special functions (which can be found in specialized monographs such as \cite{AS72} and \cite{DLMF}).

\section{Analytic continuation of the spectral \texorpdfstring{$\zeta$}{zeta}-function}\label{cont}

The process of analytic continuation of the spectral $\zeta$-function, which will allow us to extend the expression obtained 
in Theorem \ref{t3} from the strip $1/2<\Re(s)<1$ to the region of the complex plane to the left of $\Re(s)=1/2$, has been described
numerous times in the literature, most recently, for regular Sturm--Liouville operators under additional assumptions, in \cite[Subsection 3.3]{FGKS21}. 
The basic idea relies on the observation that the restriction $\Re(s)>1/2$ found for the integral representation \eqref{19a} arises
from the $z\to\infty$ behavior of the characteristic function given in\eqref{dfasymp}. This implies that in 
order to extend the expression for the spectral $\zeta$-function to a region of the complex plane to the left of $\Re(s)=1/2$, one can simply 
subtract, and then add, from the integrand in \eqref{19a}, terms of the $t\to\infty$ asymptotic expansion of 
$\ln[(te^{i\Psi})^{-m_0}F_{A,B}(te^{i\Psi})]$. This is exactly the analytic continuation process employed in 
\cite{FGKS21, FGK}. To describe this in more details, we introduce the following definition:

\begin{definition}\label{defNsmooth}
We call a $($quasi-$)$regular $($bounded from below$)$ Sturm--Liouville problem satisfying $(pr),(pr)'/r \in AC_{loc}((a,b))$ and $(pr)\big|_{(a,b)}>0$ such that the Liouville transformed potential $V$ given in \eqref{28} satisfies $V(\xi)\in C^N([\cA,\cB])$ for some $N\in\bbN_0$, \textit{$N$-smooth $($quasi-$)$regular}.
\end{definition}

Using Definition \ref{defNsmooth}, we now describe some of our previous results and illustrate when they can be easily extended to the quasi-regular setting. It was proven in \cite[Subsection 3.2]{FGKS21} that by using a Liouville transformation (see \eqref{26}--\eqref{28}) one can recast a regular Sturm--Liouville differential equation satisfying $(pr),\ (pr)'/r\in AC([a,b])$ and for some $\varepsilon >0$, $pr \geq \varepsilon$ 
on $[a,b]$
into a regular Schr\"odinger equation that is particularly suitable for the asymptotic analysis of its solutions. Furthermore, given an $N$-smooth regular problem,
the explicit large-$z$ asymptotic expansion (to order depending on $N$) of the fundamental set of solutions and, hence, of the characteristic function, can be found by employing the standard Liouville--Green (or WKBJ) method. To have the method give the correct asymptotic expansion to all orders, one must assume that the new potential is smooth on $[\cA,\cB]$.

In the quasi-regular case, one could attempt to follow the same process described above by first applying a regularization as in Theorem \ref{regularizing}. If one then assumes that the coefficients associated  with the regular problem satisfy $(PR),\ (PR)'/R\in AC([a,b])$ and for some $\varepsilon >0$, $PR \geq \varepsilon$ 
on $[a,b]$, one can proceed as in \cite[Subsection 3.2]{FGKS21}. In order to find additional terms of the asymptotic expansion, one must assume more regularity in the transformed coefficient functions. We now provide an example showing that there exist quasi-regular problems with an associated $N$-smooth regular problem (in fact, infinitely smooth).

\begin{example}
Consider the differential expression
\begin{align}\label{4.1}
\tau_{\mu}=-(x-a)^{(3\mu-1)/2}\dfrac{d}{dx}(x-a)^{(3-\mu)/2}\dfrac{d}{dx},\quad
\mu\in(0,1),\; x\in(a,b).
\end{align}
One readily verifies that $\tau_\mu$ is in the limit circle nonoscillatory case $($but not regular$)$ at $x=a$ and regular at $x=b$, and principal and nonprincipal solutions of $\tau_\mu u=0$ are given by
\begin{align}
u_{a,\mu} (0,x)=2(1-\mu)^{-1},\quad \hatt u_{a,\mu} (0,x)=(x-a)^{(\mu-1)/2},\quad
 \mu\in(0,1).
\end{align}
The generalized boundary values for $g \in \dom(T_{max,\mu})$ are then of the form
\begin{equation}
\wti g(a) =  \lim_{x \downarrow a} (x-a)^{(1-\mu)/2}g(x), \quad
\wti g^{\, \prime} (a) = 2^{-1}(1-\mu)\lim_{x \downarrow a} \big[g(x) - \wti g(a) (x-a)^{(\mu-1)/2}\big].
\end{equation}
Direct calculation yields the coefficients of the associated regular problem are
\begin{equation}
P(x)=(x-a)^{(1+\mu)/2}, \quad R(x)=(x-a)^{-(1+\mu)/2},\quad Q(x)=0,  \quad \mu\in(0,1),\; x\in(a,b).
\end{equation}
One easily verifies that $1/P, Q, R\in L^1((a,b);dx)$, $PR=1>0$ and $(PR)'/R=0$. Thus applying a Liouville transform $($see below with $k=a$$)$ results in the free Schr\"odinger equation
\begin{equation}\label{4.8}
- \frac{d^2}{d\xi^2} u (z,\xi)= zu(z,\xi), \quad \xi \in (0,\cB),\ \text{where }\ \cB=2\frac{(b-a)^{(1-\mu)/2}}{1-\mu},\; \mu\in(0,1).
\end{equation}
\end{example}

Notice that applying a Liouville transformation (choosing $k=a$ below) directly to \eqref{4.1} yields \eqref{4.8} as well. In fact, a direct calculation shows that one can obtain the potential in \eqref{28} in two equivalent ways. Namely, either by first regularizing and then performing a Liouville transformation (regardless of the regularizing function chosen), or by simply executing a Liouville transformation directly. This shows that some quasi-regular problems actually become regular simply by performing a Liouville transformation, making the initial regularization in Theorem \ref{regularizing} obsolete. This motivates the inclusion of the quasi-regular case in Definition \ref{defNsmooth}, and allows us to immediately extend the results proven in \cite[Subsection 3.2]{FGKS21} to $N$-smooth quasi-regular problems.

Let us now explicitly define a Liouville transformation (see \cite[Sect.~4.3]{Ev82}, \cite{Ev05}, and \cite[Sect. 3.5]{GNZ23}), by first assuming, in addition to Hypothesis \ref{h1}, that the 
coefficients of $\tau$ in \eqref{2.1} satisfy $(pr),(pr)'/r \in AC_{loc}((a,b))$ and $(pr)\big|_{(a,b)}>0$. One then defines
\begin{align}\label{26}
&\xi(x)=\int_k^x dt\ [r(t)/p(t)]^{1/2},\quad k\in [a,b],   \\
& \cA:=-\int_a^k dt\ [r(t)/p(t)]^{1/2}>-\infty,\quad\text{and}\quad \cB:=\int_k^b dt\ [r(t)/p(t)]^{1/2}<\infty,\\
&u(z,\xi)=[p(x(\xi))r(x(\xi))]^{1/4}y(z,x(\xi)), 
\end{align}
to recast the quasi-regular equation $\tau y(z,x)=z y(z,x)$ with $x \in (a,b)$ into the form 
\begin{align}\lb{27}
- \frac{d^2}{d\xi^2} u (z,\xi)+V(\xi)u(z,\xi)= zu(z,\xi), \quad \xi \in (\cA,\cB)\subset\R .
\end{align}
Notice that the choices $k=a$ or $k=b$ are allowed by the quasi-regular assumption, which also ensures the new interval $(\cA,\cB)$ is finite. The transformed potential $V(\xi)$ can be found to be 
\begin{align}\label{28}
\begin{split}
V(\xi)&=-\dfrac{1}{16}\dfrac{1}{p(x)r(x)}\left[\dfrac{(p(x)r(x))^{\,\prime}}{r(x)}\right]^2+\dfrac{1}{4}\dfrac{1}{r(x)} \left[\dfrac{(p(x)r(x))^{\,\prime}}{r(x)}\right]^{\,\prime} + \dfrac{q(x)}{r(x)}. 
\end{split}
\end{align}
Because of the additional conditions $(pr),(pr)'/r \in AC_{loc}((a,b))$ and $(pr)\big|_{(a,b)}>0$, the potential satisfies $V(\xi)\in L^1_{loc}((\cA,\cB);d\xi)$. This local integrability condition, however, does not prevent the potential $V(\xi)$ from diverging at the necessarily finite endpoints 
of the interval $(\cA,\cB)$. Thus without making further assumptions (as described above), the standard Liouville--Green asymptotic analysis cannot be carried out and one has to 
rely on more advanced asymptotic analysis techniques such as singular perturbation theory, matched asymptotic expansions, and multiple scale methods (see e.g. \cite{BO78,Eck79,Hol95,Lag88}). A suitable method for finding the asymptotic expansion of the fundamental system of solutions
can be selected only once a specific quasi-regular Sturm--Liouville operator is given. In other words, unlike in the $N$-smooth (quasi-)regular problems where 
an asymptotic expansion for the characteristic function can be found to higher order depending on $N$, the general regular and quasi-regular cases require, in principle, an {\it ad-hoc} asymptotic analysis for each particular form that the operator takes.   

Because of the remark above we cannot outline in detail the analytic continuation of the spectral $\zeta$-function for a general quasi-regular problem. Nevertheless, it is possible to provide important remarks about the procedure. 
According to \eqref{dfasymp} the leading term of the $z\to\infty$ asymptotic expansion of the derivative of
$\ln\,F_{A,B}(z)$ is the same for all self-adjoint extensions and depends only on the coefficients $r(x)$ and $p(x)$ of the differential expression \eqref{2.1}. The subleading terms of the expansion, however, cannot be determined by the results of Proposition \ref{Prop:lnF}. 

To elaborate on this point, according to the proof of Proposition \ref{Prop:lnF}, Weyl's asymptotic result in \eqref{WeylAsym} causes the large-$z$ asymptotic expansion of the derivative of the logarithm of the characteristic function $F_{A,B}(z)$ to be fully determined only up to a $o(z^{-1/2})$ function. This implies that the next term in the asymptotic expansion \eqref{lnF} could be a function that, as $z\to\infty$, approaches zero {\it slower} than an $O(z^{-1/2-\varepsilon})$ function with $\varepsilon>0$ arbitrarily small. For example, Proposition \ref{Prop:lnF}, with $\kappa=2$ to account for Weyl's result in \eqref{WeylAsym}, does not exclude the possibility that the next term in the expansion \eqref{lnF} is of the form $z^{-1/2}(\ln z)^{-\gamma}$, with $\gamma>0$.
While this asymptotic behavior of the logarithmic derivative of the characteristic function might be possible, we are not aware of any Sturm-Liouville problem in which it actually occurs. 
For this reason we will continue our analysis under the following  assumption.
\begin{assumption}\label{asymp0}
The derivative of $\ln\,F_{A,B}(z)$ has the large-$z$ asymptotic expansion to some order $N \in \N$ 
\begin{equation}\label{FAsym}
-\frac{d}{dz}\ln\,F_{A,B}(z)=\frac{i}{2}\cAB 
 z^{-1/2}+\sum_{n=1}^{N}c_{n}(A,B)\omega_{n}(z)+o(\omega_{N}(z)),
\end{equation}
where $\cAB = \int_a^b \sqrt{\frac{r(x)}{p(x)}} dx$ and $\omega_{n}(z)$ represents a suitable asymptotic sequence for $z\to\infty$ in any sector not containing the real line such that $\omega_{1}(z)=o(z^{-1/2-\varepsilon})$, with $\varepsilon>0$. In addition, for $C>0$, the elements $\omega_{n}(z)$ satisfy the integrability condition 
\begin{equation}
t^{-s}\omega_{n}(t e^{i\Psi})\in L^{1}((C,\infty);dt),
\end{equation}
for $\Re(s)>s_n$, with $s_{n+1}< s_{n}$, for $n\in\N_0$, and $s_0=1/2$. 
\end{assumption}
For instance, in the case of $N$-smooth (quasi-)regular Sturm--Liouville operators the appropriate asymptotic sequence, according to the Liouville--Green approximation \cite[Eq. (3.82)]{FGKS21}, has only one scale (up to $N$) and is explicitly given by 
$\omega_n(z)=z^{-(n+1)/2},\ n\in\{1,2,\dots,N\}$. If the transformed potential is smooth on the given finite interval, then this holds for all $N\in\bbN$.

When Assumption \ref{asymp0} holds, the analytic continuation of the spectral $\zeta$-function to the left of the abscissa of convergence $\Re(s)=1/2$, is performed by subtracting, and then adding, to the integrand in \eqref{19a} the terms of the large-$z$ asymptotic expansion of $\ln\left[(te^{i\Psi})^{-m_0}F_{A,B}(te^{i\Psi})\right]$. In other words, the integrability assumption and the fact that $s_{n+1}< s_{n}$, for $n\in\N_{0}$, ensures that the more terms of the asymptotic expansion are subtracted, and then added, the more the convergence region of the integral \eqref{19a} moves to the left. This process leads to the expression 
\begin{equation}\label{32}
\zeta(s;T_{A,B})=\cZ(s;A,B)+\sum_{j=-1}^{N}\cH_{j}(s;A,B),
\end{equation}  
where 
\begin{align}\label{33}
\cZ(s;A,B)&=e^{is(\pi-\Psi)}\frac{\sin(\pi s)}{\pi}\int_{0}^{\infty}dt\,t^{-s}\Bigg\{\frac{d}{dt}\ln\Big[(te^{i\Psi})^{-m_0}F_{A,B}\left(te^{i\Psi}\right)\Big]\nonumber\\
&\quad-H(t-C)\Big[-\frac{i}{2}\cAB t^{-1/2}e^{i\Psi/2}-\frac{m_0}{t}+\sum_{n=1}^{N}c_{n}(A,B)\omega_{n}\left(te^{i\Psi}\right)\Big]\Bigg\},
\end{align}
is an entire function of $s$ for $s_N<\Re(s)<1$ with $H(t)$ denoting the Heaviside function and $C>0$ an arbitrary point.
The remaining functions $\cH_{j}(s;A,B)$ can be found to be
\begin{align}\label{34}
\cH_{-1}(s;A,B)&=-ie^{is(\pi-\Psi)}\frac{\sin(\pi s)}{\pi}\frac{\cAB e^{i\Psi/2}C^{-s+(1/2)}}{2s-1},\\ \label{35}
\cH_{0}(s;A,B)&=-m_0 e^{is(\pi-\Psi)}\frac{\sin(\pi s)}{\pi}\frac{C^{-s}}{s},\\ \label{36}
\cH_{n}(s;A,B)&=e^{is(\pi-\Psi)}c_{n}(A,B)\frac{\sin(\pi s)}{\pi}\int_{C}^{\infty}dt\,t^{-s}\omega_{n}(t e^{i\Psi}),\;\; n\in\{1,2,\dots,N\}.
\end{align}

We can provide, at this point, a few general remarks concerning the analytic continuation of the spectral $\zeta$-function, $\zeta(s;T_{A,B})$.
The term $\cH_{-1}(s;A,B)$ of the analytic continuation \eqref{32}, given in \eqref{34}, is {\it universal}, namely, it is the same for all self-adjoint extensions $T_{A,B}$ of $T_{min}$. In addition, by noting its simple meromorphic structure, one can prove the following:   
\begin{lemma}\label{pole}
Let Assumption \ref{asymp0} hold. Then the spectral $\zeta$-function, $\zeta(s;T_{A,B})$, associated with a self-adjoint extension $T_{A,B}$ of a quasi-regular Sturm--Liouville 
operator has a simple pole at the point $s=1/2$ with residue $\frac{1}{2\pi}\int_a^b \sqrt{\frac{r(x)}{p(x)}} dx$.  
\end{lemma} 
Of course, this is also true for the spectral $\zeta$-function arising from regular Sturm--Liouville problems \cite{FGKS21}. The functions $\cH_{n}(s;A,B)$ in \eqref{36} are, according to Assumption \ref{asymp0}, entire for $\Re(s)>s_n$. However, once a suitable asymptotic sequence $\omega_n(z)$ is selected to describe the large-$z$ asymptotic expansion of the logarithmic derivative of $F_{A,B}(z)$ the integrals in \eqref{36} can generally be extended as  meromorphic functions, with possible branch cuts, to $\Re(s)\leq s_n$. It follows, then, that the functions $\cH_{n}(s;A,B)$ contain the information about all the possible poles and branch points of $\zeta(s;T_{A,B})$, provided Assumption \ref{asymp0} holds to all orders $N$ with $s_n \to -\infty$ as $n \to\infty$.

In the case of $N$-smooth regular Sturm--Liouville problems, the spectral $\zeta$-function develops only simple poles in a region of the complex plane depending on $N$ and it is always regular at the origin (see \cite{FGKS21}). The same meromorphic structure arises in the case of $N$-smooth quasi-regular problems described at the beginning of this section. 

In more general situations however this is no longer true and the analytically continued expression of the spectral $\z$-function can become much more involved. In fact, the spectral $\zeta$-function associated with the Laplace operator on conic manifolds exhibits an unusual structure (see e.g. \cite{Ki08} for an outline). Other examples of non-standard spectral $\zeta$-functions, possessing poles of order higher than one and branch points, can be found,
for instance, in \cite{BKD96,CZ97,Cog06} for the generalized cone and in \cite{FLP,Ki06} for some singular Sturm--Liouville operators. We further remark that if the spectral $\zeta$-function develops 
a pole at $s=0$ then the definition of the $\zeta$-regularized functional determinant needs to be modified (see \eqref{5.21} and \eqref{5.22}). Some proposals that extend the definition to this situation can also be found, for instance, in \cite{CZ04,Les98}. Furthermore, we refer to \cite{Har22} for a detailed explanation of the relation between the $\zeta$-regularized determinant and the Fredholm determinant under certain regularity assumptions.
 
As it is clear from \eqref{36}, the structure of $\zeta(s;T_{A,B})$ depends crucially on the asymptotic sequence $\omega_n(z)$ used to construct the asymptotic expansion of $F_{A,B}(z)$. If the asymptotic sequence consists entirely of inverse powers of $z$, then the corresponding $\zeta$-function has only simple poles. If the asymptotic sequence contains, instead, logarithmic terms, then the spectral $\zeta$-function develops poles 
of higher order and possibly branch points.
In the case of general Sturm--Liouville operators the asymptotic expansion of the solutions, and hence, of the characteristic function, must be
studied with non-standard asymptotic methods. The use of these methods may lead to logarithmic terms of the large parameter, also 
known as logarithmic switchback terms which become necessary when matching expansions in different regions (see e.g. \cite{chen96,Hol14,Lag84}).
An expansion with logarithmic terms of the large parameter has been considered, for instance, in physical applications, during the analytic 
continuation of the  spectral $\zeta$-function associated with a Laplace-type operator describing the Hartle-Hawking wave function of the universe. The spectral $\zeta$-function so obtained was used for the calculation of the one-loop graviton wave function of the DeSitter universe \cite{Bar92}.   

For a wide variety of quasi-regular Sturm--Liouville operators that are of interest 
in physics and applied mathematics we can express the solutions in terms of well-known special functions for which suitable asymptotic expansions 
can be found in the literature. From these, one can obtain an asymptotic expansion for the characteristic function according to the following procedure. 

At first, two, non-vanishing, linearly independent solutions, $u_{+}(z,x)$ and $u_{-}(z,x)$ of the differential equation associated with the quasi-regular Sturm--Liouville operator $T_{A,B}$ are found in terms of special functions. 
The fundamental system of solutions $\theta(z,x,a)$ and $\phi(z,x,a)$ can, in turn, be expressed in terms of the linearly independent solutions $u_{+}(z,x)$ and $u_{-}(z,x)$ by imposing the generalized boundary conditions \eqref{4}, that is
\begin{align}\label{29}
\phi(z,x,a)&=\frac{1}{\wti W(u_{+},u_{-})}\left[\wti u_{+}(z,a)u_{-}(z,x)-\wti u_{-}(z,a)u_{+}(z,x)\right],\nonumber\\
\theta(z,x,a)&=\frac{1}{\wti W(u_{+},u_{-})}\left[\wti u^{\, \prime}_{-}(z,a)u_{+}(z,x)-\wti u^{\, \prime}_{+}(z,a)u_{-}(z,x)\right].
\end{align}
The characteristic functions $F_{\a,\b}(z)$ and $F_{\varphi,R}(z)$ in \eqref{7} and \eqref{9} can then be explicitly written in terms of $u_{+}(z,x)$ and $u_{-}(z,x)$ as
\begin{align}\label{30}
F_{\alpha,\beta}(z)&=\frac{1}{\wti W(u_{+},u_{-})}\left[U_{b}(u_{-})(z)U_{a}(u_+)(z)-U_{b}(u_+)(z)U_{a}(u_-)(z)\right],\nonumber\\
F_{\varphi,R}(z)&=\frac{1}{\wti W(u_{+},u_{-})}\big[\wti u^{\, \prime}_{-}(z,a)V_{1}(u_+)(z)-\wti u^{\, \prime}_{+}(z,a)V_{1}(u_-)(z)
\nonumber\\
&\quad+\wti u_{+}(z,a)V_{2}(u_{-})(z)-\wti u_{-}(z,a)V_{2}(u_{+})(z)\big]+e^{2i\varphi}-1.
\end{align}
These expressions allow us to evaluate the large-$z$ asymptotic expansion of $F_{\a,\b}(z)$ and $F_{\varphi,R}(z)$ from the ones for $u_{+}(z,x)$ and $u_{-}(z,x)$.   

\section{Examples of quasi-regular Sturm--Liouville operators} \label{examples}
In this section, we illustrate our approach to the analytic continuation of the spectral $\zeta$-function developed in the previous section by utilizing some specific quasi-regular Sturm--Liouville operators. For all these examples, we will also explicitly compute the spectral $\zeta$-function values at positive integers by utilizing the results of Theorem \ref{t4}.

\subsection{The generalized Bessel Equation on \texorpdfstring{$(0,b)$}{(0,b)}} \lb{sub3.1}
\hfill

We start with the generalized Bessel equation following the analysis outlined in \cite{GNS21} (see also \cite{FGKLNS21,GLNPS21}).
Let $a=0$, $b\in(0,\infty)$, and consider
\begin{equation}\label{5.1}
\tau_{\d,\nu,\g} = x^{-\d}\left[-\frac{d}{dx}x^\nu\frac{d}{dx} +\frac{(2+\d-\nu)^2\g^2-(1-\nu)^2}{4}x^{\nu-2}\right],\quad
\d>-1,\; \nu<1,\; \g\geq0,\; x\in(0,b),
\end{equation}
which is singular at the endpoint $x=0$ unless $\gamma = (1-\nu)/(2+\d-\nu)$ (since the potential, $q$ is not integrable near $x=0$) and regular at $x=b$ when $b\in(0,\infty)$. Furthermore, $\tau_{\d,\nu,\g}$ is in the limit circle case at $x=0$ if $0\leq\g<1$ and in the limit point case at $x=0$ when $\g\geq1$.

\begin{remark}\lb{r3.1}
For a detailed discussion of the limit circle/limit point dichotomy and oscillatory behavior of $\tau_{\d,\nu,\g}$ at $x=0$ with more general parameter choices, we refer to \cite[Sect. 8.4]{GNZ23}. (The case of an infinite interval is also examined there.) In particular, it is shown that $\tau_{\d,\nu,\g}$ is nonoscillatory at $x=0$ $($hence, the associated minimal operator, $T_{min,\d,\nu,\g}$, is bounded from below$)$ if $2+\d-\nu>0$ and $\g\in[0,\infty)$. A convenient choice of parameters is then $\d>-1$ and $\nu<1$  since it allows one to recover the standard Bessel differential expression by setting $\d=0=\nu$ and to have continuity in the boundary conditions when considering the regular case
$\gamma = (1-\nu)/(2+\d-\nu)$ as seen below.
\hfill$\diamond$
\end{remark}

Solutions to $\tau_{\d,\nu,\g} u=zu$ are given by (cf.\ \cite[No.~2.162, p.~440]{Ka61})
\begin{align}
y_{1,\d,\nu,\g}(z,x)&=x^{(1-\nu)/2} J_{\gamma}\big(2z^{1/2} x^{(2+\d-\nu)/2}/(2+\d-\nu)\big),\quad \g\geq0,\\
y_{2,\d,\nu,\g}(z,x)&=\begin{cases}
x^{(1-\nu)/2} J_{-\gamma}\big(2z^{1/2} x^{(2+\d-\nu)/2}/(2+\d-\nu)\big), & \g\notin\bbN_0,\\
x^{(1-\nu)/2} Y_{\g}\big(2z^{1/2} x^{(2+\d-\nu)/2}/(2+\d-\nu)\big), & \g\in\bbN_0,
\end{cases}\ \g\geq0,
\end{align}
where $J_{\mu}(\dott), Y_{\mu}(\dott)$ are the standard Bessel functions of order $\mu \in \bbR$ 
(cf.\ \cite[Ch.~9]{AS72}).

We now introduce principal and nonprincipal solutions $u_{0,\d,\nu,\g}(0, \dott)$ and $\hatt u_{0,\d,\nu,\g}(0, \dott)$ of $\tau_{\d,\nu,\g} u = 0$ at $x=0$ by
\begin{align}
\begin{split} 
u_{0,\d,\nu,\g}(0, x) &= (1-\nu)^{-1}x^{[1-\nu+(2+\d-\nu)\g]/2}, \quad \gamma \in [0,\infty),   \\
\hatt u_{0,\d,\nu,\g}(0, x) &= \begin{cases} (1-\nu)[(2+\d-\nu) \gamma]^{-1} x^{[1-\nu-(2+\d-\nu)\g]/2}, & \gamma \in (0,\infty),     \lb{3.5} \\
(1-\nu)x^{(1-\nu)/2} \ln(1/x), & \gamma =0,  \end{cases}\\
&\hspace*{4.5cm} \d>-1,\; \nu<1,\; x \in (0,b).
\end{split} 
\end{align}
The generalized boundary values for $g \in \dom(T_{max,\d,\nu,\g})$ are then of the form
\begin{align}\lb{GBC}
\wti g(0)
&= \begin{cases} \lim_{x \downarrow 0} g(x)\big/\big[(1-\nu)[(2+\d-\nu) \gamma]^{-1} x^{[1-\nu-(2+\d-\nu)\g]/2}\big], & 
\gamma \in (0,1), \\[1mm]
\lim_{x \downarrow 0} g(x)\big/\big[(1-\nu)x^{(1-\nu)/2} \ln(1/x)\big], & \gamma =0, 
\end{cases} \\
\wti g^{\, \prime} (0)
&= \begin{cases} \lim_{x \downarrow 0} \big[g(x) - \wti g(0) (1-\nu)[(2+\d-\nu) \gamma]^{-1} x^{[1-\nu-(2+\d-\nu)\g]/2}\big]\\
\qquad\qquad\big/\big[(1-\nu)^{-1}x^{[1-\nu+(2+\d-\nu)\g]/2}\big], 
& \hspace{-.2cm}\gamma \in (0,1), \lb{GBC1}\\[1mm]
\lim_{x \downarrow 0} \big[g(x) - \wti g(0) (1-\nu)x^{(1-\nu)/2} \ln(1/x)\big]\\
\qquad\qquad\big/\big[(1-\nu)^{-1}x^{(1-\nu)/2}\big], & \hspace{-.2cm} \gamma =0.
\end{cases}
\end{align}

\begin{remark}
The coefficient $(1-\nu)^{-1}$ in \eqref{3.5} has been chosen so that the generalized boundary conditions \eqref{GBC} and \eqref{GBC1} reduce in the regular case, that is, in the case $\d > -1$, $\nu < 1$, and 
$\gamma = (1-\nu)/(2+\d-\nu)$ treated in \cite{EZ78}, to the ordinary boundary conditions consisting of the value of the function and its derivative at $x=0$. 
\hfill $\diamond$ 
\end{remark}

Since we are interested in the quasi-regular problem, we now focus on the case in which $\g\in[0,1)$. The normalized (at $x=0$) fundamental system of solutions is now explicitly given by
\begin{align}
& \phi_{\d,\nu,\g}(z,x,0) = (1-\nu)^{-1}(2+\d-\nu)^\gamma\Gamma(1+\g) z^{- \gamma/2} 
y_{1,\d,\nu,\g}(z,x),\quad \gamma \in [0,1),      \\
& \theta_{\d,\nu,\g}(z,x,0) = \begin{cases} (1-\nu)(2+\d-\nu)^{-\gamma - 1} \gamma^{-1} \Gamma(1 - \gamma) 
z^{\gamma/2} y_{2,\d,\nu,\g}(z,x), & \hspace{-.1cm} \gamma \in (0,1), \\[1mm]
(1-\nu)(2+\d-\nu)^{-1} [- \pi y_{2,\d,\nu,0}(z,x)\\
\quad+(\ln z-2\ln(2+\d-\nu)+2\g_E) y_{1,\d,\nu,0}(z,x)], & \hspace{-.1cm} \gamma =0, 
\end{cases}       \\ 
& \hspace*{6.2cm} \d>-1,\; \nu<1,\; z \in \bbC, \; x \in (0,b), \no
\end{align}
where $\Gamma(\dott)$ denotes the gamma function, and $\gamma_{E}$ represents Euler's constant. As the point $x=b$ is a regular endpoint, recalling \eqref{A10}, the generalized boundary conditions at $x=b$ reduce to regular ones, so the characteristic functions given in \eqref{7} and \eqref{9} are simply evaluated at $x=b$.

We will now outline the analytic continuation of the spectral $\zeta$-function associated with the generalized Bessel operator.
It is important to point out that all the self-adjoint extensions with $\alpha=0$ give rise to a characteristic function that has 
an asymptotic expansion similar to the one obtained in the $N$-smooth regular case. This implies that the corresponding spectral $\zeta$-function can be analytically continued to $\C$ to a meromorphic function with only simple poles. This was observed for the standard Bessel equation (i.e., $\d=0=\nu$) in \cite{BS87,Ki06}.
All self-adjoint extensions with $\alpha\neq 0$ and $\gamma=0$ generate, instead, large-$z$ asymptotic expansions containing logarithmic terms which lead to spectral $\zeta$-functions that develop a branch point.
For the sake of simplicity, we will consider the separated boundary conditions with $\alpha=\pi/2$, $\beta=0$, and $\gamma=0$. 
This choice is simple enough to be analyzed with straightforward calculations but also quite interesting since the ensuing characteristic function exhibits the non-standard large-$z$ expansion just mentioned.  
In this case the characteristic function reduces to 
\begin{align}\label{ex1}
 F_{\frac{\pi}{2},0,\d,\nu,0}(z)&=\frac{2(1-\nu)b^{(1-\nu)/2}}{2+\d-\nu}\Big[\frac{\pi}{2}Y_{0}\big(2z^{1/2} b^{(2+\d-\nu)/2}/(2+\d-\nu)\big)\nonumber\\
&\quad-\big(\ln z^{1/2}-\ln(2+\d-\nu)+\g_E\big)J_{0}\big(2z^{1/2} b^{(2+\d-\nu)/2}/(2+\d-\nu)\big)\Big].
\end{align}
For the purpose of finding the large-$z$ asymptotic expansion of $F_{\frac{\pi}{2},0,\d,\nu,0}(z)$, it is convenient to rewrite \eqref{ex1} in 
terms of Hankel functions
\begin{align}\label{ex2}
F_{\frac{\pi}{2},0,\d,\nu,0}(z)&=\frac{(\nu-1)b^{(1-\nu)/2}}{2+\d-\nu}\bigg[\left(\ln\,z^{1/2}-\ln(2+\d-\nu)+\g_E+\frac{i\pi}{2}\right)H^{(1)}_{0}\big(2z^{1/2} b^{(2+\d-\nu)/2}/(2+\d-\nu)\big)\nonumber\\
&\quad+\left(\ln\,z^{1/2}-\ln(2+\d-\nu)+\g_E-\frac{i\pi}{2}\right)H^{(2)}_{0}\big(2z^{1/2} b^{(2+\d-\nu)/2}/(2+\d-\nu)\big)\bigg].
\end{align}
According to \cite[Eqs. 10.2.5 and 10.2.6]{DLMF} the Hankel function $H_{\mu}^{(2)}(z)$ is dominant compared to $H_{\mu}^{(1)}(z)$ for 
$z\to\infty$ when $\Im(z)>0$. This implies that for large values of $z$ with $\Im(z^{1/2})>0$ we have
\begin{align}\label{ex3}
\ln\, F_{\frac{\pi}{2},0,\d,\nu,0}(z)&=C_{\d,\nu}+
\ln\left(\ln\,z^{1/2}-\ln(2+\d-\nu)+\g_E-\frac{i\pi}{2}\right)-i\frac{2b^{(2+\d-\nu)/2}}{2+\d-\nu}z^{1/2}\nonumber\\
&\quad -\frac{1}{2}\ln\,z^{1/2}+\ln\left[1+\sum_{k=1}^{N}(-i)^{k}\left(\frac{2b^{(2+\d-\nu)/2}}{2+\d-\nu}\right)^{-k}\frac{a_k}{z^{k/2}}\right]+O(z^{-(N+1)/2}),
\end{align}
where $C_{\d,\nu}$ is a constant and we have used the large-$z$ asymptotic expansion for $H_{\mu}^{(2)}(z)$ provided in \cite[Eq. 10.17.6]{DLMF} with 
\begin{equation}
a_{k}=\frac{\Gamma((1/2)+k)^{2}}{(-2)^{k}\pi k!}.
\end{equation} 

From \eqref{ex3} is not difficult to obtain the large-$z$ asymptotic expansion needed for the analytic continuation of the spectral $\zeta$-function, that is
\begin{align}\label{ex4}
\frac{d}{dt}\ln\, F_{\frac{\pi}{2},0,\d,\nu,0}\left(te^{i\Psi}\right)&=
\left[2t\left(\ln\, t^{1/2}-\ln(2+\d-\nu)+\g_E-\frac{i}{2}(\pi-\Psi)\right)\right]^{-1}\nonumber\\
&\quad-i\frac{b^{(2+\d-\nu)/2}}{2+\d-\nu}e^{i\Psi/2}t^{-1/2} -\frac{1}{4}t^{-1}-\sum_{j=1}^{N}\frac{j\bar{a}_{j}}{2}t^{-(2+j)/2}+O(t^{-(N+3)/2}),
\end{align}
where the coefficients $\bar{a}_{j}$ are determined through the formal asymptotic expansion by equating like powers of $t$
\begin{equation}
\ln\left[1+\sum_{k=1}^{\infty}(-i)^{k}\left(\frac{2b^{(2+\d-\nu)/2}}{2+\d-\nu}\right)^{-k}a_ke^{-ik\Psi/2}t^{-k/2}\right]
=\sum_{j=1}^{\infty}\bar{a}_{j}t^{-j/2}.
\end{equation}
In particular, the coefficients can be found by employing the fact
\begin{align}\lb{4.7}
    \ln\left(1+\sum_{m=1}^\infty c_m y^m\right) = \sum_{m=1}^\infty d_m y^m, \quad 0\leq |y|\text{ sufficiently small},
\end{align}
where
\begin{equation}\lb{4.9}
d_1=c_1, \quad d_j=c_j-\sum_{\ell=1}^{j-1} (\ell /j) c_{j-\ell}d_{\ell},\quad j \in \bbN, \; j\geq 2.
\end{equation}

One can explicitly see in \eqref{ex4} that the large-$z$ asymptotic expansion contains a logarithmic term which is responsible for the appearance of a branch point of the spectral $\zeta$-function at $s=0$.  By subtracting, and then adding, $N$ terms of the asymptotic expansion \eqref{ex4}
to the integrand in \eqref{19a} we obtain (cf. \cite[Sect. 3.3]{FGKS21})
\begin{equation}\label{ex5}
\zeta(s;T_{\frac{\pi}{2},0,\d,\nu,0})=\cZ(s)-e^{is(\pi-\Psi)}\frac{\sin(\pi s)}{\pi}\Bigg[i\left(\frac{2b^{(2+\d-\nu)/2}}{2+\d-\nu}\right)\frac{e^{i\Psi/2}}{2s-1}+\frac{1}{4s}+\sum_{j=1}^{N}\frac{j\bar{a}_{j}}{2s+j}-e^{2s\mu}\textrm{E}_{1}(2s\mu)\Bigg],
\end{equation} 
where 
\begin{align}
\cZ(s)&=e^{is(\pi-\Psi)}\frac{\sin(\pi s)}{\pi}\int_{0}^{\infty}dt\,t^{-s}\Bigg\{\frac{d}{dt}\ln\,F_{\frac{\pi}{2},0,\d,\nu,0}\left(te^{i\Psi}\right)-H(t-1)\Bigg[\left[2t\left(\ln\, t^{1/2}+\mu\right)\right]^{-1}\nonumber\\
&\quad-i\frac{b^{(2+\d-\nu)/2}}{2+\d-\nu}e^{i\Psi/2}t^{-1/2} -\frac{1}{4}t^{-1}-\sum_{j=1}^{N}\frac{j\bar{a}_{j}}{2}t^{-(2+j)/2}\Bigg]\Bigg\},
\end{align}
is entire for $-(N+1)/2<\Re(s)<1$, $\textrm{E}_{1}(z)$ is the exponential integral function (see \cite[Sect. 6.2]{DLMF}) and we have introduced, for brevity, the notation $\mu=-\ln(2+\d-\nu)+\g_E-[i(\pi-\Psi)/2]$.
The result \eqref{ex5} renders the structure of the spectral $\zeta$-function manifest. As we expect, according to Lemma \ref{pole}, 
$\zeta(s;T_{\frac{\pi}{2},0,\d,\nu,0})$ has a simple pole at $s=1/2$. In addition, simple poles appear at the points $s=-(2j-1)/2$, $j\in\N$, which is also possible in the typical smooth regular case. In a neighborhood of $s=0$ we have \cite[Eq. 6.6.2]{DLMF}
\begin{equation}
\textrm{E}_{1}(2s\mu)=-\ln(2s\mu)-\g_E-\sum_{n=1}^{\infty}\frac{(-1)(2s\mu)^{n}}{n\,n!},
\end{equation}
and, hence, we can conclude that
\begin{equation}
\zeta(s;T_{\frac{\pi}{2},0,\d,\nu,0})=-s\,\ln\,s+O(1),\quad s\to0.
\end{equation}
The last result implies that in this case the spectral $\zeta$-function of the generalized Bessel operator with $\gamma=0$ develops a branch point at the origin as has already been observed in \cite{Ki06} for the special case $\delta=0=\nu$ in \eqref{5.1}. In particular, this means the $\zeta$-regularized functional determinant $\det(T_{\frac{\pi}{2},0,\d,\nu,0})=\exp(-\frac{d}{ds}\zeta(s;T_{\frac{\pi}{2},0,\d,\nu,0})|_{s=0})$ is not defined in this case. Following \cite{Ki06}, one can instead introduce a notion of regularized determinant by utilizing a modified $\z$-function
\begin{equation}\label{5.21}
\zeta_{\textrm{reg}}(s;T_{\frac{\pi}{2},0,\d,\nu,0}):=\zeta(s;T_{\frac{\pi}{2},0,\d,\nu,0})+s\,\ln s,
\end{equation}
which has a well-defined derivative at $s=0$. In terms of \eqref{5.21}, one can define the regularized determinant as
\begin{equation}\label{5.22}
{\det}_{\textrm{reg}}(T_{\frac{\pi}{2},0,\d,\nu,0}):=\exp(-\tfrac{d}{ds}\zeta_{\textrm{reg}}(s;T_{\frac{\pi}{2},0,\d,\nu,0})|_{s=0}).
\end{equation}
In particular, a straightforward application of the analysis in \cite[Sect. 9]{Ki06} yields the generalization of \cite[Thm. 1.6]{Ki06} to the generalized Bessel equation considered here.

We now present a computation of the value of the spectral $\zeta$-function associated with the generalized Bessel operator at positive integers by exploiting Theorem \ref{t4}.
The small-$z$ asymptotic expansion of $F_{\a,\b}(z)$ and $F_{\varphi,R}(z)$, needed in order to apply Theorem \ref{t4} to this particular example, can be obtained by using the small-$z$ asymptotic expansions of the Bessel functions $J_{\g}(\dott)$ and $Y_{0}(\dott)$ (see \cite[Eq. 9.1.13]{AS72} for $Y_{0}(\dott)$). By denoting with $H_{k}$ the $k$-th harmonic number \cite[Eq. 25.11.33]{DLMF},
one finds
\begin{align}
\phi_{\d,\nu,\g}(z,b,0)&=\sum_{j=0}^\infty \dfrac{(-1)^jb^{(2+\d-\nu)j+[1-\nu+\g(2+\d-\nu)]/2}\Gamma(1+\g)}{(1-\nu)(2+\d-\nu)^{2j}j!\Gamma(\gamma+j+1)}z^j,  \no  \\
\phi^{[1]}_{\d,\nu,\g}(z,b,0)&=\sum_{j=0}^\infty \dfrac{(-1)^j[2(2+\d-\nu)j+1-\nu+\g(2+\d-\nu)]\Gamma(1+\g)}{2(1-\nu)(2+\d-\nu)^{2j}j!\Gamma(\gamma+j+1)}  \no  \\
&\hspace*{1.2cm} \times b^{(2+\d-\nu)j+[\nu-1+\g(2+\d-\nu)]/2} z^j,    \no \\
\theta_{\d,\nu,\g}(z,b,0)&=\begin{cases}
\displaystyle \sum_{j=0}^\infty \dfrac{(-1)^j (1-\nu)b^{(2+\d-\nu)j+[1-\nu-\g(2+\d-\nu)]/2}\Gamma(1-\g)}{\g (2+\d-\nu)^{2j+1}j!\Gamma(j+1-\g)}z^j,\\[2mm]
\hspace*{6.85cm}\g\in(0,1),\\[2mm]
\displaystyle \sum_{j=0}^\infty \dfrac{(-1)^j (1-\nu)[2H_j-(2+\d-\nu)\ln(b)]b^{(2+\d-\nu)j+(1-\nu)/2}}{ (2+\d-\nu)^{2j+1}(j!)^2}z^j,\\[2mm]
\hspace*{8.4cm}\g=0,
\end{cases}  \no    \\
\theta^{[1]}_{\d,\nu,\g}(z,b,0)&=
\begin{cases}
\displaystyle \sum_{j=0}^\infty \dfrac{(-1)^j[2(2+\d-\nu)j+1-\nu-\g(2+\d-\nu)]\Gamma(1-\g)}{2\g(1-\nu)^{-1}(2+\d-\nu)^{2j+1}j!\Gamma(j+1-\g)}  \\[2mm]
\qquad \times b^{(2+\d-\nu)j+[\nu-1-\g(2+\d-\nu)]/2} z^j,\quad \g\in(0,1),\\[2mm]
\displaystyle \sum_{j=0}^\infty \bigg(\dfrac{[2(2+\d-\nu)j+1-\nu][2H_j-(2+\d-\nu)\ln(b)]}{ 2(2+\d-\nu)}-1\bigg)\\[2mm] \qquad \times \dfrac{(-1)^j(1-\nu)b^{(2+\d-\nu)j+(\nu-1)/2}}{(2+\d-\nu)^{2j}(j!)^2} z^j,\quad \g=0,\\[2mm]
\end{cases}  \no\\
&\hspace{4cm} \d>-1,\; \nu<1,\; \gamma \in [0,1),\; z\in\bbC,   \lb{3.15}
\end{align}
The expansions for the quasi-derivatives can be obtained by simply differentiating term-by-term the original expansions. 

Although term-by-term differentiation of the asymptotic expansion of a function does not provide, in general, an asymptotic expansion for its derivative, this procedure is justified in this case since the asymptotic series is given through a convergent power series $($see, for example, \cite[Ch.~1, Section 8]{Ol97}$)$.
Due to length of the equations involved, we only give $\zeta(1;T_{A,B,0,0,0})$ here, that is, we briefly focus on the case when $\d=\nu=\g=0$. For separated boundary conditions 
one then obtains
\begin{align}
\begin{split}
a_{0}&=\frac{1}{2b^{1/2}}\big[2b\cos(\b)\left(\cos(\a)+\sin(\a)\ln(b)\right)-\sin(\b)\left(\cos(\a)+(2+\ln(b))\sin(\a)\right)\big],\\
a_{1}&=\frac{b^{3/2}}{8}\big[\sin (\b ) (5\cos(\a )+\sin(\a )(5\ln(b)-3))-2b\cos(\b )(\cos(\a )+\sin(\a )(\ln(b)-1))\big],\\
a_{2}&=\frac{b^{7/2}}{256} \big[2b\cos(\b)(2\cos(\a )+\sin(\a)(2\ln(b)-3))+\sin(\b)(\sin(\a)(23-18\ln(b))-18\cos(\a ))\big].
\end{split}
\end{align}
If $T_{\a,\b,0,0,0}$ does not have a zero eigenvalue, then $a_{0}\neq 0$ and, hence, one finds according to \eqref{25},
\begin{equation}\lb{3.17}
\zeta(1;T_{\a,\b,0,0,0})=\frac{b^2 [2 b \cos (\b ) (\cos (\a )+\sin (\a ) (\ln (b)-1))+\sin (\b ) (\sin (\a ) (3-5 \ln (b))-5 \cos (\a ))]}{8 b \cos (\b ) (\cos (\a )+\sin (\a ) \ln (b))-4 \sin (\b ) (\cos (\a )+\sin (\a ) (\ln (b)+2))}.
\end{equation}
If, instead, $T_{\a,\b,0,0,0}$ has a zero eigenvalue then $a_{0}=0$ and one finds
\begin{equation}\lb{3.18}
\zeta(1;T_{\a,\b,0,0,0})=
\frac{b^2 [2 b \cos (\b ) (2 \cos (\a )+\sin (\a ) (2 \ln (b)-3))+\sin (\b ) (\sin (\a ) (23-18 \ln (b))-18 \cos (\a ))]}{64 b \cos (\b ) (\cos (\a )+\sin (\a ) (\ln (b)-1))-32 \sin (\b ) (5 \cos (\a )+\sin (\a ) (5 \ln (b)-3))}.
\end{equation} 

For the case of coupled boundary conditions one finds, instead,
\begin{align}
\begin{split}
a_{0}&=-\frac{e^{i\varphi}}{2 b^{1/2}}\left[\ln(b) (R_{12}-2b R_{22})-2b R_{21}-4 b^{1/2} \cos (\varphi )+R_{11}+2R_{12}\right],\\
a_{1}&=\frac{e^{i\varphi}b^{3/2}}{8} \left[\ln (b) (5 R_{12}-2 b R_{22})-2 b R_{21}+2 b R_{22}+5R_{11}-3R_{12}\right],\\
a_{2}&=\frac{e^{i\varphi}b^{7/2}}{256} \left[\ln(b)(4 b R_{22}-18R_{12})+4 b R_{21}-6 b R_{22}-18 R_{11}+23 R_{12}\right].
\end{split}
\end{align}
If zero is not an eigenvalue of $T_{\varphi,R,0,0,0}$, $a_{0}\neq0$ and one finds
\begin{equation}
\zeta(1;T_{\varphi,R,0,0,0})=\frac{b^2 \big[\ln (b) (5 R_{12}-2 b R_{22})-2 b( R_{21}- R_{22})+5 R_{11}-3 R_{12}\big]}{4 \left[\ln (b) (R_{12}-2 b R_{22})-2 b R_{21}-4 b^{1/2} \cos (\varphi )+R_{11}+2 R_{12}\right]}.
\end{equation}
If, on the other hand, zero is an eigenvalue of $T_{\varphi,R,0,0,0}$ with multiplicity one, then $a_{0}=0$ and 
\begin{align}
\begin{split}
\zeta(1;T_{\varphi,R,0,0,0})=\frac{b^2 \big[2 \ln (b) (9 R_{12}-2 b R_{22})-2 b(2 R_{21}-3 R_{22})+18 R_{11}-23 R_{12}\big]}{32 \big[\ln(b) (5 R_{12}-2 b R_{22})-2 b R_{21}+2 b R_{22}+5 R_{11}-3 R_{12}\big]}.
\end{split}
\end{align}

The case in which zero is an eigenvalue of $T_{\varphi,R,0,0,0}$ with multiplicity two can be analyzed with the help of the Krein--von Neumann extension (see \cite{AGMT10}).
This self-adjoint extension, denoted by $T_{K,\d,\nu,\g}$ (see \cite[Example 4.1]{FGKLNS21}), is realized by imposing
the following coupled boundary conditions: $\varphi=0$ and
\begin{align}
R_{K,\d,\nu,\g}&=
  \left( {\begin{array}{cc}
   \theta_{\d,\nu,\g}(0,b,0) & \phi_{\d,\nu,\g}(0,b,0) \\
   \theta^{[1]}_{\d,\nu,\g}(0,b,0) & \phi^{[1]}_{\d,\nu,\g}(0,b,0) \\
  \end{array} } \right)\\
&=\begin{cases}
b^{[\nu-1-(2+\d-\nu)\g]/2}\\
\quad \times\begin{pmatrix}  \dfrac{1-\nu}{(2+\d-\nu)\g}b^{1-\nu} & \dfrac{1}{1-\nu}b^{1-\nu+(2+\d-\nu)\g} \\
\dfrac{(1-\nu)^2}{2(2+\d-\nu)\g}-\dfrac{1-\nu}{2} & \left[\dfrac{1}{2}+\dfrac{(2+\d-\nu)\g}{2(1-\nu)}\right] b^{(2+\d-\nu)\g}
\end{pmatrix},     \\
\hfill \g\in(0,1), \\[1mm]
\begin{pmatrix}  (1-\nu)\ln(1/b)b^{(1-\nu)/2} & \dfrac{1}{1-\nu}b^{(1-\nu)/2} \\
\dfrac{(1-\nu)^2\ln(1/b)-2(1-\nu)}{2}b^{(\nu-1)/2} & \dfrac{1}{2}b^{(\nu-1)/2}
\end{pmatrix},
\hfill \g=0.
\end{cases}
\end{align}
Using these boundary conditions leads to the characteristic function
\begin{align}\label{3.47}
F_{K,{\d,\nu,\g}}(z)=-2\left(D_{\d,\nu,\g}(z,b)-1\right),
\end{align}
where, for brevity, we introduced $D_{\d,\nu,\g}(z,b)$ as
\begin{align}
\begin{split}\lb{3.48}
D_{\d,\nu,\g}(z,b)&=\big[\phi^{[1]}_{\d,\nu,\g}(0,b,0)\theta_{\d,\nu,\g}(z,b,0)+\theta_{\d,\nu,\g}(0,b,0)\phi^{[1]}_{\d,\nu,\g}(z,b,0)\\
&\quad\,-\phi_{\d,\nu,\g}(0,b,0)\theta^{[1]}_{\d,\nu,\g}(z,b,0)-\theta^{[1]}_{\d,\nu,\g}(0,b,0)\phi_{\d,\nu,\g}(z,b,0)\big]/2.
\end{split}
\end{align}

Utilizing the expansions \eqref{3.15} in equation \eqref{3.48} yields
\begin{align}
D_{\d,\nu,\g}(z,b)&=\begin{cases}
\displaystyle 1+\dfrac{1}{2}\sum_{j=1}^{\infty}\dfrac{(-1)^j b^{(2+\d-\nu)j}}{(2+\d-\nu)^{2j}j!}\bigg[\dfrac{\Gamma(\g)}{\Gamma(j+\g)}+\dfrac{\Gamma(-\g)}{\Gamma(j-\g)}\bigg]z^j,\quad \g\in(0,1),\\[4mm]
\displaystyle 1+\sum_{j=1}^{\infty}\dfrac{(-1)^j b^{(2+\d-\nu)j}}{(2+\d-\nu)^{2j}(j!)^2}[1-jH_j]z^j,\quad \g=0,
\end{cases}
\no \\
&\hspace*{4.5cm}\d>-1,\; \nu<1,\; z\in\bbC,
\end{align}
from which, according to \eqref{3.47}, one finally obtains
\begin{align}\label{3.52}
F_{K,\d,\nu,\g}(z)&=\begin{cases}
\displaystyle \sum_{j=2}^{\infty}\dfrac{(-1)^{j+1} b^{(2+\d-\nu)j}}{(2+\d-\nu)^{2j}j!}\bigg[\dfrac{\Gamma(\g)}{\Gamma(j+\g)}+\dfrac{\Gamma(-\g)}{\Gamma(j-\g)}\bigg]z^j,\quad \g\in(0,1),\\[4mm]
\displaystyle 2\sum_{j=2}^{\infty}\dfrac{(-1)^j b^{(2+\d-\nu)j}}{(2+\d-\nu)^{2j}(j!)^2}[jH_j-1]z^j,\quad \g=0,
\end{cases}
\no \\
&\hspace*{4.1cm}\d>-1,\; \nu<1,\; z\in\bbC,
\end{align}
since one has for $j=1$, $H_{1}-1=0$ and
\begin{align}
\frac{\Gamma(-\gamma)}{\Gamma(-\gamma+1)}+\frac{\Gamma(\gamma)}{\Gamma(1+\g)}=0.
\end{align}
Notice that $F_{K,\d,\nu,\g}(z)\sim z^{2}$ which is a behavior to be expected as the Krein--von Neumann extension does have, as we have already mentioned above, a zero eigenvalue with multiplicity two.

By applying Theorem \ref{t4} to the expansion \eqref{3.52} we find
\begin{equation}
\zeta(n;T_{K,\d,\nu,\g})=-nb_n,
\end{equation}
where 
\begin{equation}
b_{1}=
\dfrac{b^{2+\d-\nu}}{(2+\d-\nu)^2 \left(\gamma ^2-4\right)},\quad \g\in[0,1),
\end{equation}
and the coefficients $b_{j}$ are given by \eqref{24} with coefficients $a_{j}$, $j\geq 2$,
given by \eqref{3.52}.
Explicitly,
\begin{align}
\zeta(1;T_{K,\d,\nu,\g})&=\frac{1}{\big(4-\g^2\big)(2+\d-\nu)^2}b^{2+\d-\nu},  \no\\
\zeta(2;T_{K,\d,\nu,\g})&=\frac{\g^4+\g^2+10}{6\big(4-\g^2\big)^2\big(9-\g^2\big)(2+\d-\nu)^4}b^{2(2+\d-\nu)},  \no\\
\zeta(3;T_{K,\d,\nu,\g})&=\frac{\g^6+7\g^4+8\g^2+32}{4\big(4-\g^2\big)^3\big(9-\g^2\big)\big(16-\g^2\big)(2+\d-\nu)^6}b^{3(2+\d-\nu)},   \no\\
\zeta(4;T_{K,\d,\nu,\g})&=\frac{\g^{12}-271\g^{10}+995\g^8+11355\g^6+67240\g^4+66856\g^2+216704}{360\big(4-\g^2\big)^4\big(9-\g^2\big)^2\big(16-\g^2\big)\big(25-\g^2\big)(2+\d-\nu)^8}   \no \\
&\quad\; \times b^{4(2+\d-\nu)},\quad \d>-1,\; \nu<1,\; \gamma \in [0,1).
\end{align}

Following the procedure illustrated in these particular examples one can find both the analytic continuation and the values at integer points for the spectral $\zeta$-function associated with the separated and coupled self-adjoint extensions of the generalized Bessel operator for all other parameters.

\subsection{The Legendre equation on (--1,1)}
\hfill

As a second example, we consider the Legendre operator on $L^2((-1,1);dx)$ associated with the expression $\tau_{Leg} = - (d/dx) (1 - x^2) (d/dx)$, $x \in (-1,1)$, which is in the singular limit circle nonoscillatory case at both endpoints. For more details, we refer to the recent paper on the Jacobi differential equation \cite{GLPS23} (of which this is a special case) and \cite[Sect. 6.2]{GLN20} (and the extensive references therein).
Principal and nonprincipal solutions $u_{\pm 1, Leg}(0, \dott)$ and $\hatt u_{\pm 1, Leg}(0, \dott)$ of $\tau_{Leg} u = 0$ at $x=\pm 1$ 
are then given by
\begin{align}
u_{\pm 1, Leg} (0,x) = 1, \quad \hatt u_{\pm 1,Leg} (0,x) = 2^{-1} \ln((1-x)/(1+x)), \quad x \in (-1,1). 
\end{align}
The generalized boundary values for $g \in \dom(T_{max, Leg})$ are then of the form 
\begin{align}
\wti g(\pm 1)&= - (p g')(\pm 1) 
= \lim_{x \to \pm 1} g(x)\big/\big[2^{-1} \ln((1-x)/(1+x))\big] ,   \label{3.60}\\ 
\wti g^{\, \prime} (\pm 1) &= \lim_{x \to \pm 1} \big[g(x) - \wti g(\pm 1) 2^{-1} \ln((1-x)/(1+x))\big]. \label{3.61}
\end{align}

We begin by determining the solutions $\phi(z,x,-1)$ and $\theta(z,x,-1)$, using two linearly independent solutions to $\tau_{Leg} u = zu,\ z\in \bbC,$
subject to the conditions
\begin{align}\label{3.63}
\wti\theta(z,-1,-1)=\wti\phi^{\, \prime}(z,-1,-1)=1,\quad \wti\theta^{\, \prime}(z,-1,-1)=\wti\phi(z,-1,-1)=0.
\end{align}
For fixed $z\in \bbC$, the equation in $\tau_{Leg} u = zu$ is a Legendre equation of the form
\begin{align}
\big(1-x^2\big)w''(x)-2xw'(x)+\nu(\nu+1)w(x)=0,\quad x\in(-1,1),
\end{align}
see, \cite[Sect. 14.2(i)]{DLMF}, with
\begin{equation}\label{3.65}
\nu=\nu(z):= 2^{-1}\big[-1 + (1+4z)^{1/2} \big],
\end{equation}
and we agree to choose the usual square root branch.

Therefore, linearly independent solutions to $\tau_{Leg} u = zu$ are $P_{\nu(z)}(\dott)$ and $Q_{\nu(z)}(\dott)$, the Legendre functions of the first and second kind of degree $\nu(z)$, respectively $($cf., e.g., \cite[Sect. 14.2(i)]{DLMF}$)$.  In particular, we can write
\begin{equation}
\begin{split}
\phi(z,x,-1) &= c_{\phi,P}(z)P_{\nu(z)}(x) + c_{\phi,Q}(z)Q_{\nu(z)}(x),  \\
\theta(z,x,-1) &= c_{\theta,P}(z)P_{\nu(z)}(x) + c_{\theta,Q}(z)Q_{\nu(z)}(x),\quad x\in (-1,1),\, z\in \bbC,   \label{3.67}
\end{split}
\end{equation}
for an appropriate set of scalars $c_{\phi,P}(z),c_{\phi,Q}(z),c_{\theta,P}(z),c_{\theta,Q}(z)\in \bbC$.  The representation for $\phi(z,x,-1)$ in \eqref{3.67} and the initial conditions in \eqref{3.63} yield the following system of equations for the coefficients $c_{\phi,P}(z)$ and $c_{\phi,Q}(z)$ (with a similar system satisfied by the coefficients $c_{\theta,P}(z)$ and $c_{\theta,Q}(z)$):
\begin{equation}
\begin{cases}
0&= c_{\phi,P}(z)\wti P_{\nu(z)}(-1) + c_{\phi,Q}(z)\wti Q_{\nu(z)}(-1),\\
1&= c_{\phi,P}(z)\wti P_{\nu(z)}^{\, \prime}(-1) + c_{\phi,Q}(z)\wti Q_{\nu(z)}^{\, \prime}(-1).
\end{cases}
\end{equation}


In order to now find the appropriate scalar functions, we compute the boundary values of the Legendre functions at both endpoints.
The limiting behavior of $P_{\nu(z)}(x)$ as $x\downarrow -1$ can be obtained from the formula (see, \cite[p.198, eq. (8.16)]{Te96}), valid for 
$x\in(-1,1)$ and $\nu\in\C$
\begin{align}\label{3.73}
\begin{split}
P_\nu(x)&=\frac{1}{\Gamma(-\nu)\Gamma(1+\nu)}\sum_{n=0}^\infty\frac{(-\nu)_n(1+\nu)_n}{(n!)^2}2^{-n}(1+x)^n\\
&\quad\,\times[2\psi(1+n)-\psi(n-\nu)-\psi(n+1+\nu)-\ln((1+x)/2))],
\end{split}
\end{align}
where $(a)_n=\Gamma(a+n)/\Gamma(a),\ n\in\N_0$ is the Pochhammer's symbol (cf. \cite[Ch. 6]{AS72}) and $\psi(\dott) = \Gamma'(\dott)/\Gamma(\dott)$ denotes the digamma function. This expression implies for $\nu\in\bbC$,
\begin{equation}
P_{\nu}(x)\underset{x\downarrow -1}{=} \dfrac{\sin(\nu\pi)}{\pi}\bigg(2\gamma_E + 2\psi(1+\nu) + \ln\bigg( \dfrac{1+x}{2}\bigg)+\Oh((1+x)\ln(1+x))\bigg)+\cos(\nu\pi)(1+\Oh(1+x)).
\end{equation}
Thus
\begin{equation}
\wti P_{\nu(z)}(-1) = \lim_{x\downarrow -1} \frac{2P_{\nu(z)}(x)}{\ln((1-x)/(1+x))}= -\frac{2}{\pi}\sin(\nu(z)\pi),\quad z\in\C.\label{3.80}
\end{equation}
As a consequence of \eqref{3.80}, one applies \eqref{3.61} and the limiting behavior of $P_{\nu(z)}(x)$ as $x\downarrow -1$ to compute
\begin{align}
\wti P_{\nu(z)}^{\, \prime}&(-1) = \lim_{x\downarrow -1} \big[P_{\nu(z)}(x) - \wti P_{\nu(z)}(-1)2^{-1}\ln((1-x)/(1+x))\big]\notag\\
&= \lim_{x\downarrow -1} \bigg[\cos(\nu(z)\pi) + \dfrac{\sin(\nu(z)\pi)}{\pi}\bigg(2\gamma_E + 2\psi(1+\nu(z)) + \ln\bigg( \dfrac{1+x}{2}\bigg)\bigg)+\frac{\sin(\nu(z)\pi)}{\pi}\ln\bigg(\frac{1-x}{1+x}\bigg)\bigg]\notag\\
&= \cos(\nu(z)\pi) + \frac{\sin(\nu(z)\pi)}{\pi}[2\gamma_E + 2\psi(1+\nu(z))],\quad z\in\C.\label{3.81}
\end{align}
Similarly, using the limiting behavior of $P_{\nu(z)}(x)$ as $x\uparrow 1$ given in \cite[Eq. 14.8.1]{DLMF},
\begin{align}
\begin{split}
&P_{\nu(z)}(x)\underset{x\uparrow 1}{=} 1+\Oh(1-x),
\end{split}
\end{align}
implies that for $z\in \C$
\begin{align}
\wti P_{\nu(z)}(1) &= \lim_{x\uparrow 1} \frac{2P_{\nu(z)}(x)}{\ln((1-x)/(1+x))}=0.\label{3.90}
\end{align}
As a consequence of \eqref{3.90}, one applies \eqref{3.61} and the limiting behavior of $P_{\nu(z)}(x)$ as $x\uparrow 1$ to compute
\begin{equation}
\wti P_{\nu(z)}^{\, \prime}(1) = \lim_{x\uparrow 1} \big[P_{\nu(z)}(x) - \wti P_{\nu(z)}(1)2^{-1}\ln((1-x)/(1+x))\big]= \lim_{x\uparrow 1} P_{\nu(z)}(x)=1,\quad z\in \C.
\end{equation}

The behavior of the second solution is slightly more complicated, so we start with the endpoint $x=1$. The limiting behavior of $Q_{\nu(z)}(x)$ as $x\uparrow 1$ with $\nu\in\C\backslash\{-\N\}$ (cf. \cite[Eq. 14.8.3]{DLMF})
\begin{align}
\begin{split}\label{5.55}
&Q_{\nu(z)}(x)\underset{x\uparrow 1}{=} -\frac{1}{2}\bigg(2\gamma_E + 2\psi(1+\nu(z))+ \ln\bigg( \dfrac{1-x}{2}\bigg)\bigg)(1+\Oh(1-x)),
\end{split}
\end{align}
implies that 
\begin{align}
\wti Q_{\nu(z)}(1) &= \lim_{x\uparrow 1} \frac{2Q_{\nu(z)}(x)}{\ln((1-x)/(1+x))}=-1.\label{3.93}
\end{align}
As a consequence of \eqref{3.93}, one applies \eqref{3.61} and the limiting behavior of $Q_{\nu(z)}(x)$ as $x\uparrow 1$ to compute
\begin{align}
\wti Q_{\nu(z)}^{\, \prime}(1) &= \lim_{x\uparrow 1} \big[Q_{\nu(z)}(x) - \wti Q_{\nu(z)}(1)2^{-1}\ln((1-x)/(1+x))\big]\notag\\
&= \lim_{x\uparrow 1} \bigg[ -\gamma_E - \psi(1+\nu(z)) +\dfrac{1}{2}\ln(2)- \dfrac{1}{2}\ln(1-x)+\dfrac{1}{2}\ln\bigg(\frac{1-x}{1+x}\bigg)\bigg]\notag\\
&= -\gamma_E - \psi(1+\nu(z)),\quad z\in\C.
\end{align}
The limiting behavior of $Q_{\nu(z)}(x)$ as $x\downarrow -1$ can now be found by using (cf. \cite[Eq. 14.9.8]{DLMF} with $\mu=0$)
\begin{align}
Q_\nu(x)=-\cos(\nu \pi)Q_\nu(-x)-\frac{\pi\sin(\nu \pi)}{2}P_\nu(-x),
\end{align}
and the limiting behavior near $1$ of $P_{\nu(z)}(x)$ and $Q_{\nu(z)}(x)$. One obtains
\begin{align}
\begin{split}
&Q_{\nu(z)}(x)\underset{x\downarrow -1}{=} \bigg(\dfrac{\cos(\nu(z)\pi)}{2}\bigg(2\gamma_E + 2\psi(1+\nu(z)) + \ln\bigg( \dfrac{1+x}{2}\bigg)\bigg)-\dfrac{\pi}{2}\sin(\nu(z)\pi)\bigg)\\
&\hspace{4.6cm}\times(1+\Oh(1+x)),\quad z\in\C.
\end{split}
\end{align}
This behavior implies
\begin{equation}
\wti Q_{\nu(z)}(-1) = \lim_{x\downarrow -1} \frac{2Q_{\nu(z)}(x)}{\ln((1-x)/(1+x))}= -\cos(\nu(z)\pi),\quad z\in\C.\label{3.84}
\end{equation}
As a consequence, one applies \eqref{3.61} and the limiting behavior of $Q_{\nu(z)}(x)$ as $x\downarrow -1$ to compute
\begin{align}
&\wti Q_{\nu(z)}^{\, \prime}(-1) = \lim_{x\downarrow -1} \big[Q_{\nu(z)}(x) - \wti Q_{\nu(z)}(-1)2^{-1}\ln((1-x)/(1+x))\big]\notag\\
&\qquad = \lim_{x\downarrow -1} \bigg[ \dfrac{\cos(\nu(z)\pi)}{2}\bigg(2\gamma_E + 2\psi(1+\nu(z)) + \ln\bigg( \dfrac{1+x}{2}\bigg)\bigg)-\dfrac{\pi}{2}\sin(\nu(z)\pi)+\dfrac{\cos(\nu(z)\pi)}{2}\ln\bigg(\frac{1-x}{1+x}\bigg)\bigg]\notag\\
&\qquad= \cos(\nu(z)\pi)[\gamma_E + \psi(1+\nu(z))]-\dfrac{\pi}{2}\sin(\nu(z)\pi),\quad z\in\C.
\end{align}

These calculations allow one to conclude that
\begin{align}
\wti P_{\nu(z)}(-1)\wti Q_{\nu(z)}^{\, \prime}(-1) - \wti P_{\nu(z)}^{\, \prime}(-1)\wti Q_{\nu(z)}(-1)=1,
\end{align}
and hence,
\begin{align}
\begin{split}
c_{\phi,P}(z) &= \cos(\nu(z)\pi),\\
c_{\phi,Q}(z) &= -2\pi^{-1}\sin(\nu(z)\pi),\\
c_{\theta,P}(z) &= \cos(\nu(z)\pi)[\gamma_E + \psi(1+\nu(z))]-\pi2^{-1}\sin(\nu(z)\pi),\\ c_{\theta,Q}(z) &= -\cos(\nu(z)\pi) - \pi^{-1}\sin(\nu(z)\pi)[2\gamma_E + 2\psi(1+\nu(z))],\quad z\in\C.
\end{split}
\end{align}
Substituting into \eqref{3.67}, one finds for $x\in (-1,1)$ and $z\in\C$
\begin{align}\label{3.88}
\notag \phi(z,x,-1) &= \cos(\nu(z)\pi)P_{\nu(z)}(x) -2\pi^{-1}\sin(\nu(z)\pi) Q_{\nu(z)}(x),\\
\theta(z,x,-1) &= \big(\cos(\nu(z)\pi)[\gamma_E + \psi(1+\nu(z))]-\pi2^{-1}\sin(\nu(z)\pi)\big)P_{\nu(z)}(x)\\
\notag &\quad\,-\big(\cos(\nu(z)\pi) + 2\pi^{-1}\sin(\nu(z)\pi)[\gamma_E + \psi(1+\nu(z))]\big)Q_{\nu(z)}(x).
\end{align}
Finally, applying the boundary values at $x=1$ provides the values needed for the construction of the characteristic function, that is for $z\in\C$,
\begin{align}\label{3.95}
\notag\wti\phi(z,1,-1) &= 2\pi^{-1}\sin(\nu(z)\pi),\\
\notag\wti\theta(z,1,-1) &= \cos(\nu(z)\pi) + 2\pi^{-1}\sin(\nu(z)\pi)[\gamma_E + \psi(1+\nu(z))],\\
\wti\phi^{\, \prime}(z,1,-1) &= \cos(\nu(z)\pi)+2\pi^{-1}\sin(\nu(z)\pi)[\gamma_E + \psi(1+\nu(z))],\\
\notag\wti\theta^{\, \prime}(z,1,-1) &=-2^{-1}\pi\sin(\nu(z)\pi)+ 2\cos(\nu(z)\pi)[\gamma_E + \psi(1+\nu(z))] \\
\notag&\quad\,+ 2\pi^{-1}\sin(\nu(z)\pi)[\gamma_E + \psi(1+\nu(z))]^2.
\end{align}
The characteristic functions given in \eqref{7} and \eqref{9} can now be written using these boundary values.
We point out that these boundary values can also be computed using hypergeometric function representations for the solutions using the results given in \cite[Appendices A--C]{GLPS23}, in particular, choosing $\alpha=0=\beta$ in \cite[Eqs. (C.13), (C.14), and (C.16)]{GLPS23} yields \eqref{3.95} after applying appropriate identities.

We now analyze the spectral $\zeta$-function according to the process outlined in Section \ref{cont}. Similarly to the case of 
the generalized Bessel operator, the separated extension characterized by $\alpha=\beta=0$, namely the Friedrichs extension, produces a
characteristic function that has a standard asymptotic expansion and, hence, leads to a meromorphic extension of the corresponding 
spectral $\zeta$-function that is similar to the one obtained in the $N$-smooth regular case. Since we are interested in showing more exotic 
behaviors of the $\zeta$-function, we will focus our attention to the set of separated self-adjoint extensions determined by the parameters
$\alpha=0$ and $0<\beta<\pi$. In this particular case we obtain
\begin{equation}\label{ex6}
F_{0,\b}(z)=-\sin(\b)\wti\phi^{\, \prime}(z,1,-1)+\cos(\b)\wti\phi(z,1,-1).
\end{equation}  
Because of the relation \eqref{3.65}, $\nu(z)\to\infty$ as $z\to\infty$. This implies that we can construct the large-$z$ asymptotic
expansion of \eqref{ex6} by first finding its large-$\nu(z)$ expansion and then by expanding the ensuing expression for large $z$ by using the relation \eqref{3.65}. 
For $\nu(z)$ large and $\Im(z^{1/2})>0$ we write
\begin{equation}\label{ex7}
\ln\,F_{0,\b}(z)=-i\nu(z)\pi-\ln(i\pi)+\ln\left[-\frac{i\pi}{2}\sin(\b)-\cos(\b)+\sin(\b)[\g_E+\psi(1+\nu(z))]\right]+O\left(e^{2i\nu(z)\pi}\right).
\end{equation}

By exploiting the asymptotic expansion for $\nu(z)\to\infty$ (cf. \cite[Eqs. 5.5.2 and 5.11.2]{DLMF})
\begin{equation}\label{ex8}
\psi\left(1+\nu(z)\right)=\ln\,\nu(z)+\frac{1}{2\nu(z)}-\sum_{k=1}^{N}\frac{B_{2k}}{2k(\nu(z))^{2k}}+O\left(\nu(z)^{-2N-2}\right),
\end{equation}
and the large-$z$ expansion of $\nu(z)$
\begin{equation}\label{ex9}
\nu(z)=z^{1/2}-\frac{1}{2}+\sum_{n=1}^{N}\frac{\sqrt{\pi}z^{-n+(1/2)}}{2^{2n-1}(n-1)!\Gamma\left((5/2)-n\right)}+O\left(z^{-n-(1/2)}\right),
\end{equation}
which can be obtained from the relation \eqref{3.65}, we find
\begin{equation}\label{ex10}
\psi\left(1+\nu(z)\right)=\ln\,z^{1/2}+\sum_{n=1}^{N}C_{n}z^{-n}+O(z^{-N-1}),
\end{equation}
where 
\begin{equation}
C_{1}=\frac{1}{6},\quad C_{2}=-\frac{1}{30},\quad C_{3}=\frac{4}{315},\quad C_{4}=-\frac{1}{105},\quad C_{5}=\frac{16}{1155},
\quad C_{6}=-\frac{1528}{45045},\quad \ldots
\end{equation}
The coefficients of higher order terms can be found with the help of a simple algebraic computer program.
By substituting the expansions \eqref{ex9} and \eqref{ex10} in \eqref{ex7} we obtain, where $D$ is a constant,
\begin{align}\label{ex11}
\ln\,F_{0,\b}(z)&=D-i\pi z^{1/2}+\ln\left[(-i(\pi/2)+\g_E)\sin(\b)-\cos(\b)+\sin(\b)\ln\,z^{1/2}\right]\nonumber\\
&\quad+\ln\left[1+\sin(\b)\left[(-i(\pi/2)+\g_E)\sin(\b)-\cos(\b)+\sin(\b)\ln\,z^{1/2}\right]^{-1}\sum_{n=1}^{N}C_{n}z^{-n}\right]\nonumber\\
&\quad-\sum_{n=1}^{N}\frac{i\pi^{3/2}}{2^{2n-1}(n-1)!\Gamma\left((5/2)-n\right)}z^{-n+(1/2)}+O\left(z^{-n-(1/2)}\right).
\end{align}

By further expanding the third logarithmic term on the right-hand side of \eqref{ex11} for large values of $z$, we arrive at the expression
\begin{align}\label{ex12}
\ln\,F_{0,\b}(z)&=D-i\pi z^{1/2}+\ln\left[(-i(\pi/2)+\g_E)\sin(\b)-\cos(\b)+\sin(\b)\ln\,z^{1/2}\right]\nonumber\\
&\quad-\sum_{n=1}^{N}\frac{i\pi^{3/2}}{2^{2n-1}(n-1)!\Gamma\left((5/2)-n\right)}z^{-n+(1/2)}+\sum_{n=1}^{N}A_{n}(z)z^{-n}+O\left(z^{-n-(1/2)}\right),
\end{align}
where, for  $n\in\N$,
\begin{equation}\label{ex13}
A_{n}(z)=(-1)^{n+1}\sum_{k=1}^{n}\omega_{nk}\sin^{k}(\beta)\left[(-i(\pi/2)+\g_E)\sin(\b)-\cos(\b)+\sin(\b)\ln\,z^{1/2}\right]^{-k},
\end{equation}
with $\omega_{nk}\in\R$ obtained according to the formal asymptotic expansion by equating like powers of $z$
\begin{equation}\label{ex14}
\ln\left[1+\sin(\b)\left[(-i(\pi/2)+\g_E)\sin(\b)-\cos(\b)+\sin(\b)\ln\,z^{1/2}\right]^{-1}\sum_{n=1}^{\infty}C_{n}z^{-n}\right]=\sum_{n=1}^{\infty}A_{n}(z)z^{-n}.
\end{equation}
The first few functions $A_{n}(z)$ are 
\begin{align}\label{ex15}
A_{1}(z)&=\frac{\sin (\beta)}{6\left[(-i(\pi/2)+\g_E)\sin(\b)-\cos(\b)+\sin(\b)\ln\,z^{1/2}\right]},\\
A_{2}(z)&=-\frac{\sin ^2(\beta )}{72 \left[(-i(\pi/2)+\g_E)\sin(\b)-\cos(\b)+\sin(\b)\ln\,z^{1/2}\right]^2}\nonumber\\
&\quad-\frac{\sin (\beta )}{30\left[(-i(\pi/2)+\g_E)\sin(\b)-\cos(\b)+\sin(\b)\ln\,z^{1/2}\right]},\\
A_{3}(z)&=\frac{\sin ^3(\beta )}{648\left[(-i(\pi/2)+\g_E)\sin(\b)-\cos(\b)+\sin(\b)\ln\,z^{1/2}\right]^3}\nonumber\\
&\quad+\frac{\sin ^2(\beta )}{180\left[(-i(\pi/2)+\g_E)\sin(\b)-\cos(\b)+\sin(\b)\ln\,z^{1/2}\right]^2}\nonumber\\
&\quad+\frac{4 \sin (\beta )}{315\left[(-i(\pi/2)+\g_E)\sin(\b)-\cos(\b)+\sin(\b)\ln\,z^{1/2}\right]}.
\end{align} 

From the expansion \eqref{ex12} one can then obtain  
\begin{align}\label{ex16}
\frac{d}{dt}\ln\,F_{0,\b}\left(te^{i\Psi}\right)&=-\frac{i\pi}{2}e^{i\Psi/2}t^{-1/2}+\sum_{n=1}^{N}\frac{i\pi^{3/2}(n-(1/2))}{2^{2n-1}(n-1)!\Gamma\left((5/2)-n\right)}t^{-n-(1/2)}e^{-i(n-(1/2))\Psi}\nonumber\\
&\quad+\sum_{j=0}^{N}B_{j}(t)t^{-j-1}e^{-ij\Psi}+O\left(t^{-N-3/2}\right),
\end{align}
where we have introduced the notation $\lambda=[\g_E-(i/2)(\pi-\Psi)]\sin(\b)-\cos(\b)$, and the functions $B_{j}(t)$ are defined as
\begin{align}\label{ex17}
B_{0}(t)&=\frac{\sin(\b)}{2\left(\lambda+\sin(\b)\ln\,t^{1/2}\right)},\\
B_{j}(t)&=(-1)^{j}\sum_{k=1}^{j}\omega_{jk}\sin^{k}(\b)\left[\frac{k\sin(\b)}{2}\left(\lambda+\sin(\b)\ln\,t^{1/2}\right)^{-k-1}+j\left(\lambda+\sin(\b)\ln\,t^{1/2}\right)^{-k}\right],\; j\in\N.\nonumber
\end{align}

By subtracting, and then adding, $N$ terms of the asymptotic expansion \eqref{ex16}
to the integrand in \eqref{19a} we arrive at the expression (cf. \cite[Sect. 3.3]{FGKS21})
\begin{align}\label{ex18}
\zeta(s;T_{0,\beta})&=\cF(s)+e^{is(\pi-\Psi)}\frac{\sin(\pi s)}{\pi}\Bigg[-i\frac{\pi e^{i\Psi/2}}{2s-1}+\sum_{n=1}^{N}\frac{i\pi^{3/2}(n-(1/2))}{2^{2n-1}(n-1)!\Gamma\left((5/2)-n\right)}\frac{e^{-i(n-(1/2))\Psi}}{s+n-(1/2)}\nonumber\\
&\quad+\sum_{n=0}^{N}e^{-i n \Psi}\int_{1}^{\infty}dt\,t^{-s-n-1}B_{n}(t)\Bigg],
\end{align}
where 
\begin{align}
    \cF(s)&=e^{is(\pi-\Psi)}\frac{\sin(\pi s)}{\pi}\int_{0}^{\infty}dt\,t^{-s}\Bigg\{\frac{d}{dt}\ln\,F_{0,\b}\left(te^{i\Psi}\right)-H(t-1)\Bigg[-\frac{i\pi}{2}e^{i\Psi/2}t^{-1/2}\nonumber\\
    &\quad+\sum_{n=1}^{N}\frac{i\pi^{3/2}(n-(1/2))}{2^{2n-1}(n-1)!\Gamma\left((5/2)-n\right)}t^{-n-(1/2)}e^{-i(n+(1/2))\Psi}+\sum_{j=0}^{N}B_{j}(t)t^{-j-1}e^{-ij\Psi}\Bigg]\Bigg\},
\end{align}
is entire for $-(2N+1)/2<\Re(s)<1$. The integral involving the functions $B_{j}(t)$ on the right-hand side of \eqref{ex18} can be computed according to the formula
\begin{equation}\label{ex19}
\int_{1}^{\infty}dt\,t^{-s-n-1}\left(\lambda+\sin(\b)\ln\,t^{1/2}\right)^{-k}=\frac{2\lambda^{1-k}}{\sin(\b)}e^{\frac{2(s+n)\lambda}{\sin(\b)}}E_{k}\left(\frac{2(s+n)\lambda}{\sin(\b)}\right),\quad n,k\in\N,
\end{equation}
which can be obtained by performing a change of variables $x=\lambda+\sin(\b)\ln\,t^{1/2}$ and by then using the definition of the 
generalized exponential integral function $E_{k}(z)$ in \cite[Eq. 8.19.3]{DLMF}. 
By utilizing \eqref{ex19} we find the expression for the spectral $\zeta$-function valid for $-(N+1)<\Re(s)<1$
\begin{align}\label{ex20}
\lefteqn{\zeta(s;T_{0,\beta})=\cF(s)+e^{is(\pi-\Psi)}\frac{\sin(\pi s)}{\pi}\Bigg[-i\frac{\pi e^{i\Psi/2}}{2s-1}+\sum_{n=1}^{N}\frac{i\pi^{3/2}(n-(1/2))}{2^{2n-1}(n-1)!\Gamma\left((5/2)-n\right)}\frac{e^{-i(n-(1/2))\Psi}}{s+n-(1/2)}}\nonumber\\
&\quad+e^{\frac{2s\lambda}{\sin(\b)}}E_{1}\left(\frac{2s\lambda}{\sin(\b)}\right)\\
&\quad+\sum_{n=1}^{N}e^{-i n \Psi}\sum_{k=1}^{n}\frac{(-1)^{k}2\omega_{nk}\lambda^{-k}}{k\sin(\b)}e^{\frac{2(s+n)\lambda}{\sin(\b)}}
\left[\frac{k\sin(\b)}{2}E_{k+1}\left(\frac{2(s+n)\lambda}{\sin(\b)}\right)+n\lambda E_{k}\left(\frac{2(s+n)\lambda}{\sin(\b)}\right)\right]\Bigg].\nonumber
\end{align}
This last result allows us to analyze the structure of the spectral $\zeta$-function for the self-adjoint extension $T_{0,\beta}$.
In line with the result of Lemma \ref{pole}, $\zeta(s;T_{0,\beta})$ has a simple pole at $s=1/2$. Moreover, simple poles appear also at the 
points $s=-(2j-1)/2$ with $j\in\N$ as in the Bessel example. The more exotic behavior of the 
spectral $\zeta$-function comes from the generalized exponential integral functions in \eqref{ex20}.

In fact, as $s\to-j$, one has \cite[Eq. 8.19.8]{DLMF}
\begin{equation}
E_{k}\left(\frac{2(s+j)\lambda}{\sin(\b)}\right)=\frac{(-1)^{k}}{(k-1)!}\left(\frac{2(s+j)\lambda}{\sin(\b)}\right)^{k-1}\ln\left(\frac{2(s+j)\lambda}{\sin(\b)}\right)+O(1),
\end{equation}   
which implies that the spectral $\zeta$-function develops branch points at $s=-j$ for every $j\in\N_0$! 
Notice that while the previous Bessel example included a branch point halfway between the first two simple poles (i.e., at  $s=0$) of the spectral $\zeta$-function, the current example effectively adds a branch point halfway between every successive pair of simple poles. As far as we are aware this is the first time that this remarkable behavior of the spectral $\zeta$-function has been observed. We further point out that one can regularize this example in the sense of Theorem \ref{regularizing} to see that the associated regular problem still has this remarkable behavior (see \cite[Ex. 8.3.1]{Ze05}). Similarly to the generalized Bessel example, one can introduce a regularized $\zeta$-function and determinant for this example via \eqref{5.21} and \eqref{5.22}.

We focus our attention, now, to the computation of the value of the spectral $\zeta$-function at positive integers. The small-$z$ asymptotic expansion of $F_{\a,\b}(z)$ and $F_{\varphi,R}(z)$, needed in order to apply Theorem \ref{t4} to this particular example, can be obtained by exploiting the small-$z$ asymptotic expansion of the digamma function $\psi(1+z)$ which can be found for instance in \cite[Eq. 6.3.14]{AS72} to be
\begin{align}
\psi(1+z)=-\gamma_E+\sum_{k=2}^\infty (-1)^k\zeta(k)z^{k-1}.
\end{align}

For the sake of brevity, we provide an explicit expression for $\zeta(1;T_{A,B})$ since it only involves the first few coefficients of the small-$z$ expansion. In the case of separated boundary conditions 
one obtains
\begin{align}
\begin{split}
a_{0}&=-\sin (\a)\cos (\b)-\cos (\a) \sin (\b),\quad a_{1}=2\cos(\a)\cos(\b)-\dfrac{\pi^2}{6}\sin(\a)\sin(\b),\\
a_{2}&=\cos (\a) \left(\frac{\pi^2}{6} \sin (\b)-2 \cos (\b)\right)+\frac{1}{6} \sin (\a) \left(\pi ^2 \sin (\b)+\pi ^2 \cos (\b)-12 \zeta(3) \sin (\b)\right).
\end{split}
\end{align}
If $T_{\a,\b}$ does not have a zero eigenvalue, then $a_{0}\neq 0$ and, hence, one finds from \eqref{25},
\begin{equation}
\zeta(1;T_{\a,\b})=\dfrac{12\cos(\a)\cos(\b)-\pi^2\sin(\a)\sin(\b)}{6(\sin (\a)\cos (\b)+\cos (\a) \sin (\b))}.
\end{equation}
If, instead, $T_{\a,\b}$ has a zero eigenvalue then $a_{0}=0$ and one finds
\begin{equation}
\zeta(1;T_{\a,\b})=
\dfrac{\cos (\a) \left(\pi ^2 \sin (\b)-12 \cos (\b)\right)+\sin (\a) \left(\pi ^2 \sin (\b)+\pi ^2 \cos (\b)-12 \zeta(3) \sin (\b)\right)}{\pi^2\sin(\a)\sin(\b)-12\cos(\a)\cos(\b)}.
\end{equation} 

In the case of coupled boundary conditions one finds
\begin{align}
\begin{split}
a_{0}&=-e^{i\varphi}(R_{11}+R_{22})+e^{2i\varphi}+1,\quad
a_{1}=e^{i\varphi}\left(-\dfrac{\pi^2}{6}R_{12}+2R_{21}\right),\\
a_{2}&=\dfrac{e^{i\varphi}}{6}\big[\pi^2R_{11}+\big(\pi^2-12\zeta(3)\big)R_{12}-12R_{21}+\pi^2R_{22}\big].
\end{split}
\end{align}
Once again, if zero is not an eigenvalue of $T_{\varphi,R}$, $a_{0}\neq0$ and one finds
\begin{equation}
\zeta(1;T_{\varphi,R})=\dfrac{e^{i\varphi}\left(\pi^2R_{12}-12R_{21}\right)}{-6e^{i\varphi}(R_{11}+R_{22})+6e^{2i\varphi}+6}.
\end{equation}
If, on the other hand, zero is an eigenvalue of $T_{\varphi,R}$ with multiplicity one, then $a_{0}=0$ and 
\begin{align}
\begin{split}
\zeta(1;T_{\varphi,R})=\dfrac{\pi^2R_{11}+\big(\pi^2-12\zeta(3)\big)R_{12}-12R_{21}+\pi^2R_{22}}{\pi^2R_{12}-12R_{21}}.
\end{split}
\end{align}

As already mentioned in the generalized Bessel example, the case in which zero is an eigenvalue with multiplicity two can be analyzed by considering the extension realized by imposing on the Legendre equation the following coupled boundary conditions: $\varphi=0$ and, by noting $\nu(0)=0$ and $\psi(1)=-\gamma_E$ in \eqref{3.95},
\begin{align}
R'=
  \left( {\begin{array}{cc}
   \wti\theta(0,1,-1) & \wti\phi(0,1,-1) \\
   \wti\theta^{\, \prime}(0,1,-1) & \wti\phi^{\, \prime}(0,1,-1) \\
  \end{array} } \right)=
  \left( {\begin{array}{cc}
  1 & 0 \\
  0 & 1 \\
  \end{array} } \right).
\end{align}
Thus one concludes the curious fact that $R=I_2$ so that this extension coincides with the periodic extension.

One easily obtains for $z\in\C$
\begin{align}\label{KVN}
\begin{split}
F_{0,R'}(z)&=-\wti\theta(z,1,-1)-\wti\phi^{\, \prime}(z,1,-1)+2=-2\wti\theta(z,1,-1)+2\\
&=-2\cos(\nu(z)\pi) -4\pi^{-1}\sin(\nu(z)\pi)[\gamma_E + \psi(1+\nu(z))]+2.
\end{split}
\end{align}
The small-$z$ expansion of \eqref{KVN} is easily found to be 
\begin{align}
F_{0,R'}(z)\underset{z\downarrow0}{=}\frac{\pi^2}{3}z^2-\frac{2}{3}\big(\pi^2-6\zeta(3)\big)z^3+\frac{1}{60}\left(100\pi^2-\pi^4-720\zeta(3)\right)z^4+\Oh\big(z^5\big),
\end{align}
from which one clearly sees that zero is an eigenvalue of multiplicity 2 as expected. Hence
\begin{equation}
\zeta(1;T_{0,R'})=\frac{2\pi^2-12\zeta(3)}{\pi^2}=2-\frac{12}{\pi^2}\zeta(3),\quad
\zeta(2;T_{0,R'})=\frac{100\pi^2-\pi^4-720\zeta(3)}{20\big(\pi^2-6\zeta(3)\big)}.
\end{equation}
Values of $\zeta(n;T_{0,R'})$ for $n\geq 3$ can be easily found with the help of a simple computer program, with more terms of the Riemann $\zeta$-function at positive odd integers appearing as $n$ becomes larger.

Another extension of particular interest is the Friedrichs which we denote by $T_{F,Leg}=T_{0,0}$. We note that the spectrum of $T_{0,0}$ may be computed explicitly to be $\sigma(T_{0,0}) = \{n^2+n\}_{n\in \bbN}$ based on \cite[Sect.~9\,$(i)$]{Ev05}.
The characteristic function for this particular case reads
\begin{align}
F_{0,0}(z)=2\pi^{-1}\sin(\nu(z)\pi),\quad z\in\C,
\end{align}
hence $z$ is an eigenvalue when $\nu(z)=n,\ n\in\N_0$, that is, when $z=n(n+1),\ n\in\N_0$, as expected. Furthermore, one finds the series expansion
\begin{align}
F_{0,0}(z)=2\pi^{-1}\sin(\nu(z)\pi)\underset{z\downarrow0}{=}2z-2z^2+\big(4-\pi^2/3\big)z^3+\big(\pi^2-10)z^4+\Oh\big(z^5\big),
\end{align}
from which applying Theorem \ref{t4} yields
\begin{equation}
\zeta(1;T_{0,0})=1,\quad
\zeta(2;T_{0,0})=\dfrac{\pi^2}{3}-3,\quad
\zeta(3;T_{0,0})=10-\pi^2.
\end{equation}

\section{Concluding remarks}
In this paper we have considered the spectral $\zeta$-function for self-adjoint extensions of quasi-regular Sturm--Liouville operators that are bounded from below. We expressed the $\zeta$-function in terms of a contour integral of a characteristic function which implicitly determines the eigenvalues of its associated self-adjoint extension. The integral representation so constructed is valid only in a certain region of the complex plane and in order to extend it to a larger region of $\C$ to the left of the natural boundary, we exploited a well-known method involving the asymptotic expansion of the characteristic function. In fact, it is worth observing that the large-$z$ asymptotic expansion of the characteristic function allows one to extend the representation of the spectral $\zeta$-function towards the left of the complex plane, while its small-$z$ expansion is employed to analyze the $\zeta$-function in a region to the right of $\Re(s)=1$. In this work we have illustrated our main techniques with the help of two examples of quasi-regular Sturm--Liouville operators. Remarkably, the example of the Legendre operator provides a spectral $\zeta$-function exhibiting a very unusual structure involving branch points at all nonpositive integer values of $s$.        

This work represents a first step towards a complete analysis of the spectral $\zeta$-function of singular Sturm--Liouville operators and while there are certainly more topics to be explored we would like to mention some that represent a natural continuation of the ideas developed here.  

In Section \ref{cont} we have described the process of analytic continuation of the spectral $\zeta$-function. Due to the fact that there is no standard form of the asymptotic expansion of the characteristic function in the singular setting, we were able to give only some general guidelines regarding the process of analytic continuation. This is to be contrasted with the $N$-smooth (quasi-)regular case, were the general form of the asymptotic expansion of the characteristic function can be found by using the Liouville--Green approximation and a complete analytic continuation can be performed in general. It would be quite interesting to study in much more detail if there are classes of singular operators in which the asymptotic expansion of the characteristic function can be given in general. This analysis would likely heavily rely on the theory of singular asymptotic expansions or Borel summation. Until such analysis can be performed, we are limited to finding the appropriate asymptotic expansion in a case-by-case basis. 

We would like to make an additional comment. One of the assumptions we have used in Section \ref{cont} to construct the analytic continuation of the spectral $\z$-function, constraints the first term of the asymptotic sequence $\omega_1(z)$ to be an $o(z^{-1/2-\varepsilon})$ function. It would be very interesting to find, if it exists, a Sturm-Liouville problem for which this assumption on $\omega_1(z)$ in the asymptotic expansion \eqref{FAsym} does not hold for any $\varepsilon > 0$. This would be the case for Sturm-Liouville problems whose largest $n\to\infty$ subleading behavior of the eigenvalues contains, for instance, terms of the form $n^{2}(\ln\,n)^{-\gamma}$, $\gamma>0$. An asymptotic expansion of the logarithmic derivative of the characteristic function for which $\omega_1(z)=o(z^{-1/2})$ but not $o(z^{-1/2-\varepsilon})$ for any $\varepsilon > 0$ would lead to a spectral $\z$-function that might present a pole at $s=1/2$ of order higher than one or even a branch point. This would qualify as an exotic behavior of the spectral $\z$-function since all cases in literature find the rightmost pole is simple.   

The investigations performed here were limited to the quasi-regular case, that is the case in which the endpoints are limit circle nonoscillatory. The next step would be to extend the technique of analytic continuation of the spectral $\zeta$-function, developed in the previous sections, to the case in which at least one endpoint is LP, with the LP endpoint being nonoscillatory for all $\lambda\in\bbR$. Since boundary conditions are 
not allowed at an endpoint that is LP, the characteristic function cannot be expected to have 
a form that is similar to \eqref{7} and \eqref{9}. In fact, 
the construction of the characteristic function in the case of an LP endpoint will involve integrability conditions on the solutions recently investigated in \cite{PS24}. Moreover, the characteristic function for operators with trace class resolvents will no longer strictly have growth $1/2$. We hope to report on this particular topic in a future work.

\appendix
\section{Basics of singular Weyl--Titchmarsh--Kodaira theory}\label{appendix}

In this appendix we briefly summarize the basic notions of singular Weyl--Titchmarsh--Kodaira theory that are necessary for a description of all self-adjoint extensions of 
the minimal quasi-regular Sturm--Liouville operator. For this overview, we mainly follow  \cite[Chs.~4, 6--8]{Ze05}, \cite{GLN20} and \cite[Ch.~13]{GNZ23}. Moreover, everything in this appendix is standard and can be found, for instance, in \cite[Chs.~8, 9]{CL85}, \cite[Sects.~13.6, 13.9, 13.0]{DS88}, \cite[Ch.~III]{JR76}, \cite[Ch.~V]{Na68}, \cite{NZ92}, \cite[Ch.~6]{Pe88}, \cite[Ch.~9]{Te14}, \cite[Sect.~8.3]{We80}, \cite[Ch.~13]{We03}.

Let us consider the differential expression \eqref{2.1} and the maximal and minimal operators given in Definition \eqref{def1}. One of the most important results in the 
theory of singular Sturm--Liouville operators is Weyl's theorem, also known as Weyl alternative:
\begin{theorem}[Weyl's Alternative] \label{At1} ${}$ \\
Assume Hypothesis \ref{h1}. Then the following alternative holds$:$ \\[1mm] 
$(i)$ For every $z\in\bbC$, all solutions $u$ of $(\tau-z)u=0$ are in $\Lr$ near $a$ 
$($resp., near $b$$)$. \\[1mm] 
$(ii)$  For every $z\in\bbC$, there exists at least one solution $u$ of $(\tau-z)u=0$ which is not in $\Lr$ near $a$ $($resp., near $b$$)$. In this case, for each $z\in\bbC\bs\bbR$, there exists precisely one solution $u_a$ $($resp., $u_b$$)$ of $(\tau-z)u=0$ $($up to constant multiples$)$ which lies in $\Lr$ near $a$ $($resp., near $b$$)$. 
\end{theorem}
This theorem naturally leads to the limit point and limit circle classification of $\tau$ at an endpoint of the interval as follows
\begin{definition}  \label{LCLP}
In case $(i)$ of Theorem \ref{At1}, $\tau$ is said to be in the \textit{limit circle $($LC$)$ case} at $a$ $($resp., at $b$$)$ and is called quasi-regular at $a$ $($resp., at $b$$)$. 
\\ [1mm]
In case $(ii)$ of Theorem \ref{At1}, $\tau$ is said to be in the \textit{limit point $($LP$)$ case} at $a$ $($resp., at $b$$)$. \\[1mm]
If $\tau$ is in the limit circle case at $a$ and $b$ then $\tau$ is also called \textit{quasi-regular} on $(a,b)$. 
\end{definition}
The LP and LC classification of an endpoint depends entirely on the particular form of the functions $p(x)$, $q(x)$, and $r(x)$ comprising the differential expression 
$\tau$ (consult \cite[Ch. 7]{Ze05} for the relevant results).
Before proceeding with the characterization of the self-adjoint extensions in this singular case, one needs to be sure that such extensions indeed exist. To this end, one introduces the notion of deficiency indices.
\begin{definition}  
The positive and negative deficiency spaces of $T_{min}$ are given, respectively, by
\begin{align}
\cD_{+}=\{f\in\dom(T_{max})|(T_{max} - i I)f=0\},\quad \cD_{-}=\{f\in\dom(T_{max})|(T_{max} + i I)f=0\}.
\end{align}
The positive integers 
\begin{align}
n_\pm(T_{min}) &= \dim(\cD_{\pm}) 
\end{align}  
are called the positive and negative deficiency indices of $T_{min}$.
\end{definition}
It is clear, from this definition, that the positive deficiency index $n_{+}$ is simply equal to the number of solutions $f\in\Lr$ of the equation $(\tau-iI)f=0$  
(and similarly for $n_{-}$). According to the general theory of Sturm--Liouville operators, $n_{\pm}\leq 2$ and the existence of self-adjoint extensions of $T_{min}$ is 
determined by the following:
\begin{theorem}  
The minimal operator $T_{min}$ has self-adjoint extensions if and only if $n_{-}=n_{+}$. 
\end{theorem}
In the case of quasi-regular Sturm--Liouville operators considered in this work, the endpoints of the interval are limit circle and the following ensures the existence of self-adjoint extensions of $T_{min}$.
\begin{theorem}  
If $\tau$ is limit circle at both $a$ and $b$, then $n_{-}=n_{+}=2$. 
\end{theorem} 
Once the existence of self-adjoint extensions has been established, the next goal is to obtain their explicit representation. This can be accomplished with the introduction 
of {\it generalized boundary values}.
\begin{definition}\label{defigeneral}  
Let $v_j \in \dom(T_{max})$, $j=1,2$, satisfy 
\begin{equation}
W(\ol{v_1}, v_2)(a) = W(\ol{v_1}, v_2)(b) = 1, \quad W(\ol{v_j}, v_j)(a) = W(\ol{v_j}, v_j)(b) = 0, \; j= 1,2.  
\end{equation}
$($E.g., real-valued solutions $v_j$, $j=1,2$, of $(\tau - \lambda) u = 0$ with $\lambda \in \bbR$, such that 
$W(v_1,v_2) = 1$.$)$ For $g\in\dom(T_{max})$, its generalized boundary values are   
\begin{align}
\begin{split} 
\wti g_1(a) &= - W(v_2, g)(a), \quad \wti g_1(b) = - W(v_2, g)(b),    \\
\wti g_2(a) &= W(v_1, g)(a), \quad \;\,\,\, \wti g_2(b) = W(v_1, g)(b).   
\end{split} 
\end{align}
\end{definition}
The generalized boundary values of a function in the maximal domain allow for a representation of the self-adjoint extensions of $T_{min}$ that mimics exactly the one 
obtained in the case of regular Sturm--Liouville operators. In fact one can prove the following, which is familiar from the regular setting.   
\begin{theorem}\label{extensions}
$(i)$ All self-adjoint extensions $T_{\al,\be}$ of $T_{min}$ with separated boundary conditions are of the form
\begin{align}
& T_{\al,\be} f = \tau f, \quad \al,\be\in[0,\pi),   \notag\\
& f \in \dom(T_{\al,\be})=\big\{g\in\dom(T_{max}) \, \big| \, \wti g_1(a)\cos(\al)+ \wti g_2(a)\sin(\al)=0;   \label{A1} \\ 
& \hspace*{5.5cm} \, \wti g_1(b)\cos(\be)- \wti g_2(b)\sin(\be) = 0 \big\}.    \notag
\end{align}
$(ii)$ All self-adjoint extensions $T_{\varphi,R}$ of $T_{min}$ with coupled boundary conditions are of the form
\begin{align}
\begin{split}\label{A2} 
& T_{\varphi,R} f = \tau f,    \\
& f \in \dom(T_{\varphi,R})=\bigg\{g\in\dom(T_{max}) \, \bigg| \begin{pmatrix} \wti g_1(b)\\ \wti g_2(b)\end{pmatrix} 
= e^{i\varphi}R \begin{pmatrix}
\wti g_1(a)\\ \wti g_2(a)\end{pmatrix} \bigg\},
\end{split}
\end{align}
where $\varphi\in[0,\pi)$, and $R \in SL(2,\bbR)$ .  \\[1mm] 
$(iii)$ Every self-adjoint extension of $T_{min}$ is either of type $(i)$ or of type 
$(ii)$.
\end{theorem}
We would like to point out that if the endpoints $a$ and $b$ are regular, then the generalized boundary values reduce to the ordinary boundary values (under appropriate choices of $v_j$ in \eqref{A8}) and the self-adjoint 
extensions described in Theorem \ref{extensions} become the self-adjoint extensions of a regular Sturm--Liouville operator in $\Lr$. 
As we have mentioned in Section \ref{s2}, this work is concerned mainly with quasi-regular operators that are bounded from below.
\begin{definition}\label{bound}
Let $\lambda_0 \in \bbR$. Then $T_{min}$ is called bounded from below by $\lambda_0$, 
and one writes $T_{min} \geq \lambda_0 I$, if 
\begin{equation} 
(u, [T_{min} - \lambda_0 I]u)_{L^2((a,b);rdx)}\geq 0, \quad u \in \dom(T_{min}).
\end{equation}
\end{definition}
We would like to mention that if $T_{min}$ is bounded from below, then all of its symmetric extensions are also bounded from below, and
for this class of operators one can 
introduce the notion of principal and nonprincipal solutions of $\tau u = \lambda u$ for appropriate $\lambda \in \bbR$. We will show that principal and nonprincipal solutions provide us with a representation of the generalized boundary values that is quite convenient when performing
explicit calculations.        
We start by reviewing some oscillation theory with particular emphasis on principal and nonprincipal solutions, a notion originally due to Leighton and Morse \cite{LM36} (see also Rellich \cite{Re43}, \cite{Re51} and Hartman and Wintner \cite[Appendix]{HW55}) (see also \cite{CGN16}, \cite[Sects.~13.6, 13.9, 13.0]{DS88}, 
\cite[Ch.~XI]{Ha02}, \cite{NZ92}, \cite[Chs.~4, 6--8]{Ze05}).
\begin{definition}
Assume Hypothesis \ref{h1}. Fix $c\in (a,b)$ and $\lambda\in\bbR$. Then $\tau - \lam$ is
called {\it nonoscillatory} at $a$ $($resp., $b$$)$, 
if every real-valued solution $u(\lambda,\dott)$ of 
$\tau u = \lambda u$ has finitely many
zeros in $(a,c)$ $($resp., $(c,b)$$)$. Otherwise, $\tau - \lam$ is called {\it oscillatory}
at $a$ $($resp., $b$$)$.
\end{definition} 
The following result ensures that when $T_{min}$ is bounded from below then $\tau - \lam$ is nonoscillatory at both endpoints (and vice-versa). 
More precisely   
\begin{theorem} \label{At2} 
The following items $(i)$--$(iii)$ are
equivalent$:$ \\[1mm] 
$(i)$ $T_{min}$ $($and hence any symmetric extension of $T_{min})$
is bounded from below. \\[1mm] 
$(ii)$ There exists a $\nu_0\in\bbR$ such that for all $\lambda < \nu_0$, $\tau - \lam$ is
nonoscillatory at $a$ and $b$. \\[1mm]
$(iii)$ For fixed $c, d \in (a,b)$, $c \leq d$, there exists a $\nu_0\in\bbR$ such that for all
$\lambda<\nu_0$, $\tau u = \lambda u$ has $($real-valued\,$)$ nonvanishing solutions
$u_a(\lambda,\dott) \neq 0$,
$\hatt u_a(\lambda,\dott) \neq 0$ in $(a,c]$, and $($real-valued\,$)$ nonvanishing solutions
$u_b(\lambda,\dott) \neq 0$, $\hatt u_b(\lambda,\dott) \neq 0$ in $[d,b)$, such that 
\begin{align}
&W(\hatt u_a (\lambda,\dott),u_a (\lambda,\dott)) = 1,
\quad u_a (\lambda,x)=\oh(\hatt u_a (\lambda,x))
\text{ as $x\downarrow a$,} \label{A3} \\
&W(\hatt u_b (\lambda,\dott),u_b (\lambda,\dott))\, = 1,
\quad u_b (\lambda,x)\,=\oh(\hatt u_b (\lambda,x))
\text{ as $x\uparrow b$,} \\
&\int_a^c dx \, p(x)^{-1}u_a(\lambda,x)^{-2}=\int_d^b dx \, 
p(x)^{-1}u_b(\lambda,x)^{-2}=\infty,\label{A4} \\
&\int_a^c dx \, p(x)^{-1}{\hatt u_a(\lambda,x)}^{-2}<\infty, \quad 
\int_d^b dx \, p(x)^{-1}{\hatt u_b(\lambda,x)}^{-2}<\infty. \label{A5}
\end{align}
\end{theorem}
The existence of solutions with certain growth properties, (\ref{A4}) and \eqref{A5} as guaranteed by item $(iii)$ in the above theorem, allows us to define principal
and nonprincipal solutions as follows:
\begin{definition}
Suppose that $T_{min}$ is bounded from below, and let 
$\lambda\in\bbR$. Then $u_a(\lambda,\dott)$ $($resp., $u_b(\lambda,\dott)$$)$ in Theorem \ref{At2}\,$(iii)$ is called a {\it principal} $($or {\it minimal}\,$)$
solution of $\tau u=\lambda u$ at $a$ $($resp., $b$$)$. A real-valued solution 
$v_a(\lambda,\dott)$ $($resp., $v_b(\lambda,\dott)$$)$ of $\tau
u=\lambda u$ linearly independent of $u_a(\lambda,\dott)$ $($resp.,
$u_b(\lambda,\dott)$$)$ is called {\it nonprincipal} at $a$ $($resp., $b$$)$.
\end{definition}
Once the existence of principal and nonprincipal solutions has been established, one needs an efficient procedure to construct them.  
First, we have the following:
\begin{lemma} \label{l2.11} The functions $u_a(\lambda,\dott)$ and $u_b(\lambda,\dott)$ in Theorem
\ref{At2}\,$(iii)$ are unique up to nonvanishing real constant multiples. Moreover,
$u_a(\lambda,\dott)$ and $u_b(\lambda,\dott)$ are minimal solutions of
$\tau u=\lambda u$ in the sense that 
\begin{align}
u(\lambda,x)^{-1} u_a(\lambda,x)&=\oh(1) \text{ as $x\downarrow a$,} \\ 
u(\lambda,x)^{-1} u_b(\lambda,x)&=\oh(1) \text{ as $x\uparrow b$,}
\end{align}
for any other solution $u(\lambda,\dott)$ of $\tau u=\lambda u$
$($nonvanishing near $a$, resp., $b$$)$ with
$W(u_a(\lambda,\dott),u(\lambda,\dott))\neq 0$, respectively, 
$W(u_b(\lambda,\dott),u(\lambda,\dott))\neq 0$.
\end{lemma}

The following lemma provides an explicit expression for principal and nonprincipal solutions.
\begin{lemma}
Let $u(\lambda,\dott) \neq 0$ be any nonvanishing solution of $\tau u=\lambda
u$ near $a$ $($resp., $b$$)$. Then for $c_1>a$ $($resp., $c_2<b$$)$
sufficiently close to $a$ $($resp., $b$$)$, 
\begin{align}\label{nonprin}
& \hatt u_a(\lambda,x)=u(\lambda,x)\int_x^{c_1}dx' \, 
p(x')^{-1}u(\lambda,x')^{-2} \\
& \quad \bigg(\text{resp., }\hatt u_b(\lambda,x)=u(\lambda,x)\int^x_{c_2}dx' \, 
p(x')^{-1}u(\lambda,x')^{-2}\bigg) 
\end{align} 
is a nonprincipal solution of $\tau u=\lambda u$ at $a$ $($resp.,
$b$$)$. If $\hatt u_a(\lambda,\dott)$ $($resp., $\hatt
u_b(\lambda,\dott)$$)$ is a nonprincipal solution of $\tau u=\lambda u$
at $a$ $($resp., $b$$)$ then 
\begin{align}\label{prin}
& u_a(\lambda,x)=\hatt u_a(\lambda,x)\int_a^{x}dx' \, 
p(x')^{-1}{\hatt u_a(\lambda,x')}^{-2} \\
& \quad \bigg(\text{resp., } u_b(\lambda,x)=\hatt u_b(\lambda,x)\int^b_{x}dx' \, 
p(x')^{-1}{\hatt u_b(\lambda,x')}^{-2}\bigg)
\end{align} 
is a principal solution of $\tau u=\lambda u$ at $a$ $($resp., $b$$)$. 
\end{lemma}
In practice, one first finds a non-vanishing solution of $\tau u=\lambda u$, for $\lambda\in\bbR$ near one of the endpoints. This solution is then used 
in \eqref{nonprin} to construct a nonprincipal solution at that endpoint. The nonprincipal solution is, in turn, exploited, according to \eqref{prin}, to 
obtain the principal solution at that endpoint.   
 
We can now return to the topic of self-adjoint extensions. In terms of principal and nonprincipal solutions, the generalized boundary values of 
a function $g \in \dom(T_{max})$ can be rewritten as follows:
\begin{theorem} \label{At3}
Assume that $\tau$ is quasi-regular on $(a,b)$ and that $T_{min} \geq \lambda_0 I$ for some $\lambda_0 \in \bbR$. Denote by 
$u_a(\lambda_0, \dott)$ and $\hatt u_a(\lambda_0, \dott)$ $($resp., $u_b(\lambda_0, \dott)$ and 
$\hatt u_b(\lambda_0, \dott)$$)$ principal and nonprincipal solutions of $\tau u = \lambda_0 u$ at $a$ 
$($resp., $b$$)$, normalized so that
\begin{equation}
W(\hatt u_a(\lambda_0,\dott), u_a(\lambda_0,\dott)) = W(\hatt u_b(\lambda_0,\dott), u_b(\lambda_0,\dott)) = 1.
\end{equation}  
Then, for all $g \in \dom(T_{max})$, one obtains 
\begin{align}
\begin{split} 
\wti g(a) & =\wti g_1(a)=  - W(u_a(\lambda_0,\dott), g)(a)= \lim_{x \downarrow a} \f{g(x)}{\hatt u_a(\lambda_0,x)},    \\
\wti g(b) & =\wti g_1(b)=  - W(u_b(\lambda_0,\dott), g)(b)= \lim_{x \uparrow b} \f{g(x)}{\hatt u_b(\lambda_0,x)},    
\label{A6} 
\end{split} \\
\begin{split} 
{\wti g}^{\, \prime}(a) &=\wti g_2(a)= W(\hatt u_a(\lambda_0,\dott), g)(a) = \lim_{x \downarrow a} \f{g(x) - \wti g(a) \hatt u_a(\lambda_0,x)}{u_a(\lambda_0,x)},    \\ 
{\wti g}^{\, \prime}(b) &=\wti g_2(b)= W(\hatt u_b(\lambda_0,\dott), g)(b) = \lim_{x \uparrow b} \f{g(x) - \wti g(b) \hatt u_b(\lambda_0,x)}{u_b(\lambda_0,x)}.    \label{A7}
\end{split} 
\end{align}
In particular, the limits on the right-hand sides of \eqref{A6} and \eqref{A7} exist. 
\end{theorem}
It is worth pointing out that if $\tau$ is regular on the finite interval $[a,b]$, then one can choose $v_j \in \dom(T_{max})$, $j=1,2$, 
such that 
\begin{align}\label{A8}
v_1(x) = \begin{cases} \vartheta(\lambda,x,a), & \text{for $x$ near $a$}, \\
\vartheta(\lambda,x,b), & \text{for $x$ near $b$},  \end{cases}   \quad 
v_2(x) = \begin{cases} \varphi(\lambda,x,a), & \text{for $x$ near $a$}, \\
\varphi(\lambda,x,b), & \text{for $x$ near $b$},  \end{cases}   
\end{align} 
where $\varphi(\lambda,\, \cdot \,,d)$, $\vartheta(\lambda,\, \cdot \,,d)$, $d \in \{a,b\}$, are real-valued solutions of $(\tau - \lambda) u = 0$, $\lambda \in \bbR$, satisfying the boundary conditions 
\begin{align}\label{A9}
\begin{split} 
& \varphi(\lambda,a,a) = \vartheta^{[1]}(\lambda,a,a) = 0, \quad \vartheta(\lambda,a,a) = \varphi^{[1]}(\lambda,a,a) = 1, \\ 
& \varphi(\lambda,b,b) = \vartheta^{[1]}(\lambda,b,b) = 0, \quad \; \vartheta(\lambda,b,b) = \varphi^{[1]}(\lambda,b,b) = 1. 
\end{split} 
\end{align} 
It is, then, not very difficult to verify that
\begin{align}\label{A10}
\wti g_1 (a) = g(a), \quad \wti g_1 (b) = g(b), \quad \wti g_2 (a) = g^{[1]}(a), \quad \wti g_2 (b) = g^{[1]}(b),   
\end{align}
that is, the generalized boundary conditions at a regular endpoint reduce to the ordinary boundary conditions for three-coefficient regular 
Sturm--Liouville operators in $\Lr$ .

Theorem \ref{At3} implies that all the self-adjoint extensions of the minimal operator that is bounded from below are given by either \eqref{A1} (for separated ones)
or \eqref{A2} (for coupled ones) with the generalized boundary values given by \eqref{A6} and \eqref{A7}. Let us reflect, for a moment, on the practical 
advantage provided by the quasi-regular case. As it can be surmised by the Definition \eqref{defigeneral}, the generalized boundary values of $g \in \dom(T_{max})$
can be evaluated once two linearly independent maximal domain functions are found. Unfortunately, an explicit expression for such functions might be hard to find in general.
When principal and nonprincipal solutions exist, then Theorem \ref{At3} provides a formula for the generalized boundary values in terms of principal and 
nonprincipal solutions {\it for one particular} $\lambda_{0}\in\bbR$. In practice, in order to evaluate generalized boundary values, one would utilize the principal and nonprincipal solutions for $\lambda_{0}=0$ since one can often obtain their explicit expressions.

\medskip

\noindent 
{\bf Acknowledgments.} M.P. was supported by the Methusalem grant METH/21/03 -- long term structural funding of the Flemish Government and acknowledges the hospitality of the Mittag--Leffler Institute where parts of this work were written.

\medskip

 
\end{document}